\documentclass[11pt]{article}
\usepackage[margin=0.8in]{geometry}

\usepackage{graphicx}
\usepackage{framed}
\usepackage[normalem]{ulem}

\usepackage{xfrac, amsthm,thm-restate,thmtools}
\usepackage{amssymb}
\usepackage{amsbsy}
\usepackage{amsfonts}
\usepackage{enumerate}
\usepackage{etoolbox}
\usepackage{subcaption}
\usepackage{url}
\usepackage{enumitem}
\usepackage{mathtools}
\usepackage{comment}
\usepackage[font={small,it}]{caption}
\usepackage{bbm}
\usepackage[usenames,dvipsnames]{color}
\usepackage{xcolor}
\usepackage{algorithmic}
\usepackage{algorithm}
\usepackage{array}
\usepackage{nicefrac}
\usepackage{bm}

\usepackage[colorlinks=true, allcolors=black]{hyperref}
\hypersetup{
  colorlinks   = true, 
  urlcolor     = blue, 
  linkcolor    = blue, 
  citecolor   = blue 
}

\usepackage{thmtools,thm-restate}

\DeclareMathOperator*{\argmax}{arg\,max}

\newtheorem{lemma}{Lemma}

\newtheorem{corollary}{Corollary}

\newtheorem{remark}{Remark}
\newtheorem{proposition}{Proposition}

\newcommand{\fac}{\Phi}

\newcommand{\alg}{\textsc{Alg}}

\newcommand{\x}{\mathbf x}
\newcommand{\opt}{\bm{\omega}}
\newcommand{\W}{\mathcal{W}}

\title{\bfseries Universal and Tight Online Algorithms for \\ Generalized-Mean Welfare}

\author{Siddharth Barman\thanks{Indian Institute of Science. {\tt barman@iisc.ac.in}} \quad Arindam Khan\thanks{Indian Institute of Science. {\tt arindamkhan@iisc.ac.in}} \quad Arnab Maiti \thanks{Indian Institute of Technology, Kharagpur. {\tt maitiarnab9@gmail.com}}}
\date{}

\begin{document}
\maketitle

\thispagestyle{empty}
\begin{abstract}
We study fair and efficient allocation of divisible goods, in an online manner, among $n$ agents. The goods arrive online in a sequence of $T$ time periods. The agents' values for a good are revealed only after its arrival, and the online algorithm needs to fractionally allocate the good, immediately and irrevocably, among the agents. Towards a unifying treatment of fairness and economic efficiency objectives, we develop an algorithmic framework for finding online allocations to maximize the generalized mean of the values received by the agents. In particular, working with the assumption that each agent's value for the grand bundle of goods is appropriately scaled, we address online maximization of $p$-mean welfare. Parameterized by an exponent term $p \in (-\infty, 1]$, these means encapsulate a range of welfare functions, including social welfare ($p=1$), egalitarian welfare ($p \to -\infty$), and Nash social welfare ($p \to 0$).

We present a simple algorithmic template that takes a threshold as input and, with judicious choices for this threshold, leads to both universal and tailored competitive guarantees. First, we show that one can compute online a {\em single allocation} that $O (\sqrt{n} \log n)$-approximates the optimal $p$-mean welfare for all $p\le 1$. The existence of such a universal allocation is interesting in and of itself. Moreover, this universal guarantee achieves essentially tight competitive ratios for specific values of $p$. 

Next, we obtain improved competitive ratios for different ranges of $p$ by executing our algorithm with $p$-specific thresholds, e.g., we provide $O(\log ^3 n)$-competitive ratio for all $p\in (\frac{-1}{\log 2n},1)$. 

We complement our positive results by establishing lower bounds to show that our guarantees are essentially tight for a wide range of the exponent parameter.
\end{abstract}

\clearpage

\thispagestyle{empty}
\tableofcontents
\clearpage

\setcounter{page}{1}

\section{Introduction}
\label{section:introduction}

Resource-allocation settings are ubiquitous and often require the assignment of resources (goods) that arrive over time. In particular, these settings require that each good gets allocated (among the participating agents) as it arrives, and one needs to make these allocative decisions without knowing the values of the pending goods. Consider, as an example, a food bank that distributes food donations (essentially a divisible good) every day among soup kitchens (agents)~\cite{prendergast2017food, aleksandrov2015online}. Here, the perishable nature of the good (food) mandates online allocations, and supply (and demand) variability leads to limited information about the future. The online model is also applicable in scheduling contexts wherein computational resources, which become available over time, have to be shared among users \cite{blazewicz2019handbook, pinedo2012scheduling, leung2004handbook}. 

In such settings---and resource-allocation contexts, in general---economic efficiency and fairness are fundamental objectives. Motivated by these considerations and application domains, such as the ones mentioned above, a growing body of work in recent years has been directed towards the study of online fair division \cite{aleksandrov2020online}. The current paper contributes to this thread of research with a welfarist perspective. 

Specifically, we provide a unified treatment of fairness and efficiency objectives by developing an online algorithm for finding allocations that maximize the \emph{generalized mean} of the values achieved by the agents. Formally, for exponent parameter $p \in \mathbb{R}$, the $p$th generalized mean of $n$ nonnegative values $\nu_1, \nu_2, \ldots, \nu_n$ is defined as $\left(\frac{1}{n} \sum_{i=1}^n \nu_i^p \right)^{1/p}$. As a family of objective functions (parameterized by $p$), generalized means encapsulate well-studied measures of economic efficiency as well as fairness: the $p=1$ case corresponds to average social welfare (arithmetic mean), which is a standard measure of economic efficiency. At the other end of the spectrum with $p \to -\infty$, the $p$th generalized mean corresponds to egalitarian welfare (the minimum value across the agents), a fundamental measure of fairness. Furthermore, when $p$ tends to zero, the generalized mean reduces to Nash social welfare (specifically, the geometric mean)--a prominent objective which achieves a balance between the extremes of social and egalitarian welfare. Notably, $p$-means, with $p \in (-\infty, 1]$, exactly constitute a family of functions characterized by a set of natural, fairness axioms, including the Pigou-Dalton principle \cite{moulin2004fair}. Therefore, $p$-mean welfare, with $p \leq 1$, provides us with an important and axiomatically justified family of objectives.

The current work focuses on the online allocation of divisible goods, i.e., goods that can be assigned fractionally among the $n$ participating agents. The divisible goods arrive online, one by one, and upon the arrival of each good $t$, every agent $i$ reports her (nonnegative) value, $v^t_i$, for the good. At this point the online algorithm fractionally distributes the good $t$ among the agents--if agent $i$ receives $x^t_i \in [0,1]$ fraction of the good, then her value increases by $x^t_i \ v^t_i$. We have $T$ goods overall and after the goods have been allocated, each agent $i$ achieves a value $\sum_{t=1}^T x^t_i v^t_i$. Considering exponent parameters $p \leq 1$, the online algorithm's objective is to compute allocations that (approximately) maximize the $p$-mean of the agents' values. Recall that the performance of online algorithms is established in terms of their competitive ratio; in the current context,  it is the worst case ratio (over all instances) between the optimal (offline) $p$-mean welfare and the welfare of the allocation computed online. 

Along with this standard model and performance metric, we work with a scaling assumption utilized in prior work in online fair division (see, e.g., \cite{gkatzelis2020fair}, \cite{bogomolnaia2019simple}, \cite{banerjee2020online}); in particular, we assume throughout that, for every agent, the cumulative value of the entire set of goods is equal to one. Note that if the agents' valuations are arbitrarily scaled, then a sub-linear competitive ratio is not possible, even for specific values of the parameter $p<1$; for unscaled valuations, Banerjee et al.~\cite{banerjee2020online} provide an example that rules out a sub-linear competitive ratio specifically for Nash social welfare. To circumvent such overly pessimistic negative results, prior work has worked with this scaling assumption ($\sum_{t=1}^T v^t_i = 1$ for all agents $i$) and we conform to it as well. It is relevant to note that Banerjee et al.~\cite{banerjee2020online} view this assumption in the framework of \emph{algorithms with predictions}; see \cite{mitzenmacher2020algorithms} for a textbook treatment of this topic. In this framework, one assumes that the algorithm has a priori (side) information about each agent's value for the grand bundle.\footnote{The result of Banerjee et al.~\cite{banerjee2020online} is robust to prediction errors. Extending the current work along these lines is an interesting direction of future work.} Another way to realize this scaling is via scrips: upfront, each agent receives scrips (tokens) of overall value one and can distribute the tokens online to indicate her values over the goods. Note that the widely-used platform {\tt spliddit.org} \cite{goldman2015spliddit} asks for the valuations to be entered with a scaling in place, albeit in an offline manner. Overall, subject to the above-mentioned scaling, our results hold in the adversarial model wherein the value of each good $t$ can be set by an adaptive adversary based upon the fractional assignments till round $(t-1)$. \\

\noindent
{\bf Related Work.} 
Fair division has been extensively studied for over seven decades \cite{brams1996fair, brandt2016handbook, moulin2004fair}. Since a detailed discussion on (offline) algorithms for fair division is beyond the scope of the current work, we will primarily focus on prior results on online algorithms and divisible goods. A work closely related to ours is that of Banerjee et al.~\cite{banerjee2020online}, who develop an $O(\log n)$-competitive online algorithm for maximizing Nash social welfare  over divisible goods. Furthermore, Banerjee et al.~provide an intricate lower bound showing that this competitive ratio is the best possible (up to an absolute constant) for Nash social welfare. The current work obtains an $O(\log^3 n)$ competitive ratio for Nash social welfare. Our algorithmic template, however, spans all $p \leq 1$, and provides tight---up to poly-log factor---competitive guarantees for a wide range of the exponent parameter. In terms of algorithm design, Banerjee et al.~\cite{banerjee2020online} utilize the primal-dual method to obtain a competitive guarantee. Applications of this design schema, in related online settings, include the work of Devanur and Jain~\cite{devanur2012online} along with Azar et al.~\cite{azar2010allocate}. By contrast, we develop a distinctive charging argument to design a single algorithmic template for all $p \leq 1$.    

Bounding envy is another well-studied goal in the literature on online fair division; see, e.g., \cite{gkatzelis2020fair,benade2018make,bogomolnaia2019simple, zeng2020fairness}, and \cite{aleksandrov2020online} for a survey. Complementary to these results, we address $p$-mean welfare.  

The current divisible-goods model captures machine scheduling with splittable jobs~\cite{jansen2021empowering, correa2015strong, epstein2006online}. In this setup one needs to schedule $T$ jobs among $n$ machines, and each job can be split into multiple parts that get assigned to different machines. In contrast to $p$-mean welfare maximization, the focus of these scheduling results is on makespan minimization.

\subsection{Our Results.} 

In this paper we develop both upper bounds and lower bounds for online welfare maximization.  \\

\noindent
{\bf Upper Bounds.}
Our online algorithm (see Section \ref{section:algorithm}) works with a given threshold $\fac$. Setting this threshold judiciously enables us to obtain both expansive and tailored competitive guarantees. 

First, we prove that a particular choice of the threshold (specifically $\fac = 8 \sqrt{n} \log (2n)$) leads to an online algorithm that achieves a universal competitive ratio of $O\left(\sqrt{n}\ \log n\right)$ for $p$-mean welfare maximization, simultaneously for all $p \leq 1$; here $n$ is the total number of agents. 

\begin{restatable}{theorem}{UniversalUpperBound}
\label{theorem:universal-comp-ratio}
For the $p$-mean welfare maximization problem---with divisible goods and scaled valuations---one can compute online a single allocation that $O\left(\sqrt{n}\ \log n\right)$ approximates the optimal $p$-mean welfare, simultaneously for all $p\leq 1$.
\end{restatable}

Theorem \ref{theorem:universal-comp-ratio} is established in  Section \ref{section:universal}. This theorem in fact provides a novel guarantee specifically for egalitarian welfare ($p=-\infty$). Also, we note that Banerjee et al.~\cite{banerjee2020online} prove---via a direct example---that, for egalitarian-welfare maximization and any constant $\varepsilon>0$, there does not exist an online algorithm with competitive ratio $n^{1/2 - \varepsilon}$ (see also Corollary \ref{corollary:lower-bound-egalitarian} in Section~\ref{section:lower-bounds}). Hence, this lower bound ensures that, for egalitarian welfare, the competitive ratio of Theorem \ref{theorem:universal-comp-ratio} is tight, up to a log term. Moreover, the guaranteed existence of a single allocation that achieves a nontrivial $p$-mean welfare guarantee, simultaneously all $p \leq 1$, is interesting in its own right.  

We also obtain improved guarantees for different ranges of $p$. In particular, we show that executing our algorithm with $p$-specific thresholds $\fac$ leads to improved competitive ratios for a wide range of the exponent parameter $p \leq 1$. Section \ref{section:tight} details our upper bounds listed in Table \ref{table:UBLB}. Note that, for a wide range of the exponent parameter, the lower bounds in the table closely match the upper bounds. 

\begin{center}
\captionof{table}{Upper and lower bounds on the competitive ratio for $p$-mean welfare maximization. Here, the lower bounds hold for any constant $\varepsilon >0$.}
\label{table:UBLB}
\begin{tabular}{| m{3.2cm}| m{3cm} |m{3cm}  |m{3cm} | } 
\hline
Range of $p$ & Algorithm threshold $\fac$ & Upper Bound & Lower Bound \\
\hline
Egalitarian Welfare ($p=-\infty$)& $8\sqrt{n}\log(2n)$ &$O(\sqrt{n}\log n)$ \ \ \ \ (Section \ref{section:tight-egalitarian}) & $O\left(n^{1/2 - \varepsilon} \right)$ \cite{banerjee2020online} \\
\hline
Nash Social Welfare ($p=0$) & $8\log^3(2n)$ &$O(\log^3 n)$ \text{ }(Section \ref{sec:5.2}) & $O\left( \log^{1- \varepsilon} n \right)$ \cite{banerjee2020online} \\
\hline
$p\in (-\infty,-1]$& $8\sqrt{n}\log(2n)$ &$O(\sqrt{n}\log n)$ \ \ \ \ (Section \ref{sec:5.7})& $O\left(n^{\frac{|p|}{2|p|+1} - \varepsilon}\right)$ (Section \ref{sec:6.2})\\
\hline
$p\in (-1,\nicefrac{-1}{4}]$& $8 n^{\frac{|p|}{|p|+1}}\log^2(2n)$  &$O\left( n^{\frac{|p|}{|p|+1}}\log^2 n \right)$ (Section \ref{section:tight-minus-four-one}) & $O\left(n^{\frac{|p|}{2|p|+1} - \varepsilon}\right)$ (Section \ref{sec:6.2})\\
\hline
$p\in (\nicefrac{-1}{4},\frac{-1}{\log(2n)}]$& $8(2n)^{2|p|}\log^3(2n)$ &$O\left(n^{2|p|}\log^3 n\right)$ (Section \ref{sec:5.5}) & $2^{-\left( 2 + \nicefrac{2}{|p|} \right)} \cdot n^{\frac{|p|}{2|p|+1}}$ (Section \ref{sec:6.2})\\
\hline
$p\in (\frac{-1}{\log(2n)},0]$& $32\log^3(2n)$ &$O(\log^3 n)$ \text{ }(Section \ref{sec:5.4})& $ > 1$ (Section \ref{sec:6.1})\\

\hline
$p\in(0,1)$ & $16\log^3(2n)$ &$O(\log^3 n)$ \text{ }(Section \ref{sec:5.3}) & $>1$ (Section \ref{sec:6.1})\\
\hline
\end{tabular}
\end{center}

\noindent 
{\bf Lower Bounds.}
The following two theorems (proved in Section \ref{section:lower-bounds}) provide the lower bounds mentioned in Table \ref{table:UBLB}. Theorem \ref{theorem:lowerboundsubopt} shows that sub-optimality is unavoidable for $p <1$.\footnote{By contrast, for $p=1$, a greedy online algorithm (that assigns each good to the agent that values it the most) finds an allocation with maximum possible (average) social welfare.} Theorem \ref{theorem:lower-bound} provides a stronger negative result for all $p <0$. 

\begin{restatable}{theorem}{LowerBoundSubOpt}
\label{theorem:lowerboundsubopt}
For any $p <1$, the $p$-mean welfare maximization problem does not admit an online algorithm that computes optimal allocations, i.e., the competitive ratio of any online algorithm is strictly greater than one. 
\end{restatable}

\begin{restatable}{theorem}{LowerBoundGeneral} 
\label{theorem:lower-bound}
For any $p<0$, there does not exist an online algorithm with competitive ratio strictly less than $2^{-\left( 2 + \nicefrac{2}{|p|} \right)} \cdot n^{\frac{|p|}{2|p|+1}}$ for the $p$-mean welfare maximization problem. 
\end{restatable}
\section{Notation and Preliminaries}
We study the problem of allocating $T$ divisible goods among $n$ agents in an online manner. Let $[n] \coloneqq \{1,\ldots, n\}$ denote the set of agents and $[T] \coloneqq \{1,\ldots, T\}$ denote the set of goods. The divisible goods arrive online, one by one, in $T$ rounds overall. The value of good $t \in T$, for every agent $i \in [n]$, is revealed only when the good arrives (i.e., in round $t$); specifically, let $v_i^t \in \mathbb{R}_{\geq 0}$ be the (nonnegative) value of good $t \in [T]$ for agent $i \in [n]$.

For each good $t$, the online algorithm must make an irrevocable allocation decision, i.e., assign the good fractionally among the agents. Let  $x_i^t \in [0,1]$ denote the fraction of the good $t \in [T]$ assigned to agent $i \in [n]$. Note that at most one unit of any good $t$ is assigned among the agents: $\sum_{i=1}^n x^t_i \leq 1$. Furthermore, for each agent $i \in [n]$, a \emph{bundle} $x_i  \coloneqq (x_i^t)_{t\in[T]} \in [0,1]^{T}$ refers to a tuple that denotes the fractional assignments of all the $T$ goods to agent $i$. 

An \emph{allocation} $\x = (x_i)_{i\in[n]} \in {[0,1]^{n\times T}}$ refers to a fractional assignment of the goods among all the agents such that no more than one unit of any good is assigned.\footnote{The online algorithm has an allocation (of all $T$ goods) in hand only after the completion of all the $T$ rounds.}  

Throughout, we will assume that all the agents have {\em additive valuations}; in particular, for bundle $x_i  = (x_i^t)_{t\in[T]}\in [0,1]^{T}$, agent $i$'s valuation $v_i(x_i) \coloneqq \sum_{t=1}^T x^t_i  v^t_i$. 

In this work, the extent of fairness and economic efficiency of allocations is measured by the \emph{generalized means} of the values that the allocations generate among the agents. Specifically, with exponent parameter parameter {$p \in (-\infty, 1]$}, the $p$th generalized mean of $n$ nonnegative numbers $\nu_1, \nu_2, \ldots, \nu_n \in \mathbb{R}_{\geq 0}$ is defined as 
\begin{align*}
{\rm M}_p(\nu_1, \ldots, \nu_n) & \coloneqq \left( \frac{1}{n} \sum \limits_{i=1}^n \nu_i^p \right)^\frac{1}{p}.
\end{align*} 
The generalized means ${\rm M}_p(\cdot)$ with {$p  \leq 1$}, constitute a family of functions that capture multiple fairness and efficiency measures: ${\rm M}_p(\cdot)$ corresponds to the arithmetic mean when $p =1$, and as $p$ tends to zero, ${\rm M}_p$---in the limit---is equal to the geometric mean. Also, $\lim_{p \to -\infty} {\rm M}_p(\nu_1, \ldots, \nu_n) = \min \left\{\nu_1, \ldots, \nu_n \right\}$. Therefore, following standard convention, we will write ${\rm M}_0 (\nu_1, \ldots, \nu_n) \coloneqq  \left( \prod_{i=1}^n \nu_i  \right)^{1/n}$ and ${\rm M}_{-\infty} (\nu_1, \ldots, \nu_n) \coloneqq \min \left\{\nu_1, \ldots, \nu_n \right\}$. 

Considering generalized means as a parameterized family of welfare objectives, we define the \emph{$p$-mean welfare} ${\rm M}_p(\x)$ of an allocation $\x = (x_i)_{i\in[n]} \in {[0,1]}^{n\times T}$ as 
\begin{align*}
{\rm M}_p (\x) \coloneqq {\rm M}_p \left( v_1(x_1), v_2(x_2), \ldots, v_n(x_n) \right) = \left( \frac{1}{n} \sum_{i=1}^n v_i(x_i)^p \right)^{1/p}
\end{align*}
Here, ${\rm M}_1(\x)$ denotes the \emph{average social welfare} of allocation $\x$ and ${\rm M}_0(\x)$ denotes the allocation's \emph{Nash social welfare}, ${\rm M}_0(\x) =  \left( \prod_{i=1}^n v_i(x_i) \right)^{1/n}$. In addition, ${\rm M}_{-\infty}(\x)$ denotes the \emph{egalitarian welfare}, ${\rm M}_{-\infty}(\x) = \min_{i \in [n]} \ v_i(x_i)$.   

We will assume throughout that for every agent the valuation of the grand bundle of goods $[T]$ is equal to one, i.e., for each $i \in [n]$, we have $\sum_{t=1}^T v_i^t=1$. See Section \ref{section:introduction} for a discussion on this scaling assumption. Also, without loss of generality, we will assume that, for all agents $i \in [n]$ and goods $t \in[T]$, the value $v_i^t \leq \frac{1}{n^2}$. This condition can be achieved online by considering---for each good---$n^2$ identical copies of value $1/n^2$ times the value of the underlying good. Note that our results hold in the adversarial model wherein the value of a good $t$ can be set by an adaptive adversary (subject to the mentioned scaling) based upon the fractional assignments till round $(t-1)$.

\section{Online Algorithm}
\label{section:algorithm}

This section details our online algorithm, $\alg(\fac)$, for maximizing $p$-mean welfare, for $p \le 1$. The  algorithm operates with a given threshold $\fac$ and upon the arrival of each good $t$, it distributes half of the good uniformly among the agents, i.e., the algorithm initializes fractional assignment $x^t_i = \frac{1}{2} \cdot \frac{1}{n}$, for each agent $i \in [n]$; see Line 3 in the algorithm and note that this initialization satisfies $\sum_{i=1}^n x^t_i = 1/2$. The remaining half of the good is further divided into $\log(2n)$ equal parts; in particular, for each $\alpha \in \left\{ \frac{1}{2^k} : 1 \leq k \leq \log (2n) \right\}$, the algorithm sets fractional assignments $x^{t, \alpha}_i \in [0,1]$ across the agents such that $\sum_{i=1}^n  x^{t, \alpha}_i  = \frac{1}{2} \cdot \frac{1}{\log (2n)}$. Note that such fractional assignments ensure that overall one unit of the good $t$ is assigned across the agents. 

For each $\alpha \in \left\{ \frac{1}{2^k} : 1 \leq k \leq \log (2n) \right\}$, to distribute $\frac{1}{2 \log (2n)}$ fraction of the good $t$, the algorithm maintains two subsets of agents:  $A_t^\alpha$, referred to as the active set of agents, and $B_t^\alpha$, a ``vulnerable'' subset of agents. For each $\alpha$ and good $t$, an agent $i$ is said to be active and, hence, included in the $A_t^\alpha$, iff so far from the $\alpha$ part agent $i$ has received value less than $\alpha/\fac$, i.e., iff {$\sum_{s=1}^{t-1} v_i^s x^{s, \alpha}_i <\alpha/\fac$}. For any $\alpha$, only active agents continue to receive nonzero fractional assignments. Specifically, for the active agent $a$ that values the current good $t$ the most (see Line \ref{line:max-active-agent}), we set $x^{t, \alpha}_a = \frac{1}{2} \cdot \frac{1}{2 \log (2n)}$. Furthermore, the algorithm considers a subset of the active agents $B_t^\alpha \subseteq A^\alpha_t$ for whom  the pending goods cumulatively have value less than $\alpha/4$; see Line \ref{line:prop-alpha} and recall the scaling that $\sum_{s=1}^T v^s_\ell = 1$ for all agents $\ell$. The remaining ($\frac{1}{4 \log (2n)}$) part of the good is uniformly distributed among the agents in $B^\alpha_t$. At a high level, the algorithm maintains a set of active agents (who have yet to receive a sufficiently high value) and a subset of vulnerable agents (for whom limited value is left among the pending goods). The algorithm then strikes a balance between greedily assigning the good (Line \ref{line:max-active-agent}) and uniformly distributing it among the vulnerable agents (Line \ref{line:prop-alpha}).  Indeed, the algorithm is computationally efficient and conceptually simple--we consider this as a relevant contribution, since such features lend the algorithm to large-scale implementations and explainable adaptations.

\begin{algorithm}[!htbp]
	\caption{\alg($\fac$)} 
	\begin{algorithmic}[1]
\STATE Initialize index $t=1$ and, for each $\alpha \in \left\{ \frac{1}{2^k} : 1 \leq k \leq \log (2n) \right\}$, initialize sets $A^\alpha_t = [n]$ and $B^\alpha_t= \emptyset$.
\FOR{each good $t=1$ to $T$}
\STATE Initialize fractional assignment $x^t_i = \frac{1}{2n}$, for each agent $i \in [n]$ \label{line:prop}.
\FOR{all $k=1$ to $\log(2n)$}
        \STATE Set $\alpha = \frac{1}{2^k}$ and {initialize fractional assignment $x^{t,\alpha}_i=0$, for each agent $i \in [n]$}.
		\STATE Select agent $a = \argmax_{j\in A_t^\alpha} v_j^t$ and assign $x_a^{t,\alpha} = \frac{1}{4\log(2n)}$ \label{line:max-active-agent}.
		\STATE {For each agent $i \in B_t^\alpha$, update $x_i^{t,\alpha} \gets x_i^{t,\alpha} + \frac{1}{4\log(2n)} \frac{1} {|B_t^\alpha|}$} \label{line:prop-alpha}.
		\STATE Set $A_{t+1}^\alpha =  A_t^\alpha\setminus \left\{j \in [n] : \sum_{s=1}^t v_j^{s}x_j^{s ,\alpha}\geq \frac{\alpha}{\fac} \right\}$ \label{line:atalpha}.
		\STATE Set $B_{t+1}^\alpha= \left\{\ell\in A_{t+1}^\alpha: \sum_{s=1}^t v_\ell^{s} > 1-\frac{\alpha}{4} \right\}$ \label{line:btalpha}.
		\ENDFOR
		\STATE Update $x_i^t \gets x_i^t  + \sum_{\alpha} x_i^{t,\alpha}$ for all agents $i \in [n]$.
		\ENDFOR
		\RETURN allocation $\x = \left( x_i^t \right)_{i,t}$.
	\end{algorithmic}
\end{algorithm}

Write $\x = (x_i)_{i \in [n]} \in [0,1]^{n \times T}$ to denote the allocation returned by $\alg(\fac)$; here, $x_i = (x^t_i)_{t \in [T]} \in [0,1]^T$ is the bundle assigned to agent $i \in [n]$.

Also, let $\opt = (\omega_i)_{i\in[n]}$ be an arbitrary allocation wherein each agent $i \in [n]$ receives a bundle of value at least $1/2n$ but less than $1$, i.e., $\frac{1}{2n}\leq v_i(\omega_i) < 1$. Next, we establish lemmas that hold for any such allocation $\opt$ and any threshold $\fac \leq n/4$. We will instantiate these lemmas with judicious choices of threshold $\fac$ and for different values of the exponent parameter $p$; in each ($p$-specific) invocation, we will consider $\opt$ as an allocation that approximately maximizes the $p$-mean welfare.
\begin{remark}
\label{remark:near-opt}
Note that for each $p$, there exists an allocation $\opt=(\omega_i)_i$ such that $\frac{1}{2n}\leq v_i(\omega_i) <1$ and the $p$-mean welfare of $\opt$ is at least half of the optimal $p$-mean welfare:\footnote{Here, the upper bound $v_i(\omega_i) <1$ follows from the scaling assumption.} fix any $p \leq 1$, write $\widehat{\opt} = \left(\widehat{\omega}^t_i \right)_{i, t}$ to denote an allocation that maximizes the $p$-mean welfare, and set $\omega_i^t = \frac{1}{2n} + \frac{1}{2} \widehat{\omega}_i^t$, for all $i$ and $t$.
\end{remark}

At a high level, we aim to bound the number agents that---in allocation $\x$---achieve value  smaller than what they achieve in $\opt$. Upper bounding the number of such sub-optimal agents will enable us to establish $p$-mean welfare guarantees in subsequent sections. 

For each $\alpha \in \left\{ \frac{1}{2^k} : k \in [\log (2n)] \right\}$, let $H(\alpha ,\opt)$ denote the subset of agents that achieve value at least $\alpha$ in allocation $\opt$, i.e., $H(\alpha ,\opt) \coloneqq \left\{ i \in [n] : v_i(\omega_i) \geq \alpha \right\}$. Lemma \ref{commonlem:1} (below) shows that, at any point of time, at most $\frac{8n\log(2n)}{\fac}$ of such high-valued agents are included in the set $B^\alpha_t$ (populated in Line \ref{line:btalpha} {of \alg($\fac$)}).  

Complementary to the set $H(\alpha, \opt)$, we define the set of agents $L(\alpha, \opt) \coloneqq \{i\in[n] : v_i(\omega_i)<\alpha\}$.  Also, let  $\widehat{L} (\alpha, \x)$   denote the set of agents that achieve value less than $\frac{\alpha}{8 \fac}$ in $\x$, i.e., $\widehat{L} (\alpha, \x) \coloneqq \left\{i\in[n] : v_i(x_i)<\frac{\alpha}{8 \fac} \right\}$. Lemma \ref{common:lem3} (below) relates the number of such low-valued agents in the respective allocations. 

In addition, comparing the values that agents receive in these two allocations, we will consider the subset of agents, $S_\fac(\opt)$, that are $(2\fac)$-sub-optimal in $\x$; write $S_\fac (\opt) \coloneqq \left\{ i \in [n] : v_i(x_i) < \frac{1}{2\fac} v_i(\omega_i) \right\}$. 

\begin{lemma}\label{commonlem:1}
For any iteration $t \leq T$ of $\alg(\fac)$ and any $\alpha  \in \left\{ \frac{1}{2^k} : 1 \leq k \leq \log(2n) \right\}$, we have 
\begin{align*}
\left|B_t^\alpha \cap H(\alpha ,\opt) \right| \leq \frac{8n \log(2n)}{\fac}.    
\end{align*}
\end{lemma}
\begin{proof}
Fix any iteration (good) $t$ and $\alpha \in \left\{ \frac{1}{2^k} : 1 \leq k \leq \log(2n) \right\}$. Note that, if agent $i \in B^\alpha_t$, then $\sum_{s=t+1}^T v_i^{s} \leq \alpha/4$; see Line \ref{line:btalpha} {in \alg($\fac$)} and recall that the valuations are scaled to satisfy $\sum_{t'=1}^T v^{t'}_i = 1$. Furthermore, for each agent $i \in H(\alpha ,\opt)$, by definition, the value $v_i(\omega_i) = \sum_{s=1}^T \omega_i^s v_i^s \geq \alpha$. These inequalities imply that, for each agent $i \in H(\alpha ,\opt) \cap B^\alpha_t$, we have  
\begin{align}
\sum_{s=1}^t \omega_i^s v_i^s \geq \frac{3 \alpha}{4} \label{ineq:highenough}
\end{align}

Given that agent $i$ is active during iteration $t$ (specifically, $i \in B^\alpha_t \subseteq A^\alpha_t$), agent $i$ must have been active ($i \in A^\alpha_s$) during all previous iterations $s \leq t$. Next, write $a^{\alpha}_s$ to denote the agent that was selected among active agents in iteration $s$ (see Line \ref{line:max-active-agent} {in \alg($\fac$)}), $a^\alpha_s \coloneqq \argmax_{j \in A_{s}^\alpha} v_j^{s}$. The fact that agent $i \in A^\alpha_s$ gives us $v^s_i \leq v^s_{a^\alpha_s}$ for each $s \leq t$. Using this bound and equation (\ref{ineq:highenough}) we get
\begingroup
\allowdisplaybreaks
\begin{align}
\frac{3 \alpha}{4} \left| B_t^\alpha\cap H(\alpha ,\opt) \right| & \leq \sum_{i \in B_t^\alpha \cap H(\alpha ,\opt)} \left( \sum_{s=1}^t\omega_i^{s}v_i^{s} \right) \tag{via (\ref{ineq:highenough})} \nonumber \\
& = \sum_{s=1}^t \ \sum_{i\in B_t^\alpha\cap H(\alpha ,\opt)} \omega_i^{s} \ v_i^{s} \nonumber \\
&\leq \sum_{s=1}^t \ \sum_{i \in B_t^\alpha\cap H(\alpha ,\opt)} \omega_i^{s} \  v_{a^\alpha_s}^{s} \tag{since $v_i^{s} \leq v^s_{a^\alpha_{s}}$} \nonumber \\
& \leq \sum_{s=1}^t v_{a^\alpha_s}^{s}  \tag{since $\sum_i \omega_i^{s}\leq 1$ for each $s$} \nonumber \\
& \leq \sum_{s=1}^t  4\log (2n) \ x_{a^\alpha_s}^{s,\alpha} \ v_{a^\alpha_s}^{s} \tag{since $x_{a^\alpha_s}^{s,\alpha}\geq \frac{1}{4\log (2n)}$; Line \ref{line:max-active-agent} in \alg($\fac$)} \nonumber \\
& = 4 \log (2n) \sum_{s=1}^t  x_{a^\alpha_s}^{s,\alpha} \ v_{a^\alpha_s}^{s} \nonumber \\
&  \leq 4 \log (2n) \sum_{s=1}^t \left( \sum_{j=1}^n x_j^{s,\alpha} \ v_j^{s} \right) \nonumber \\
&  = 4 \log (2n) \sum_{j=1}^n \ \sum_{s=1}^t x_j^{s,\alpha} \ v_j^{s} \label{ineq:social-welf}
\end{align}
\endgroup

We next bound the right-hand-side of the previous inequality by showing that $\sum_{s=1}^t x_j^{s,\alpha} \ v_j^{s} \leq \frac{3\alpha}{2\fac}$, for all agents $j \in [n]$. Note that if $j \in A^\alpha_t$, then $\sum_{s=1}^t x_j^{s,\alpha} \ v_j^{s} \leq \frac{\alpha}{\fac}$. Otherwise, if $j \in [n] \setminus A^\alpha_t$, then $j$ was removed from the active set $A^\alpha_{r}$ during some iteration $r \leq t$ and we have $\sum_{s=1}^{t} x_j^{s,\alpha} \ v_j^{s} \leq \sum_{s=1}^{r-1} x_j^{s,\alpha}\ v_j^{s} + v^{r}_j \leq \frac{\alpha}{\fac}+\frac{1}{n^2}$; the last inequality follows from the fact that the goods have value at most $1/n^2$. Since $\fac\leq n/4$ and $\alpha \geq  \frac{1}{2n}$, we get, for all agents  $j \in [n]$:
\begin{align}
    \sum_{s=1}^t x_j^{s,\alpha} \ v_j^{s} \leq \frac{\alpha}{\fac}+\frac{1}{n^2} \leq \frac{3 \alpha}{2 \fac} \label{ineq:allagents}
\end{align} 

Equations (\ref{ineq:social-welf}) and (\ref{ineq:allagents}) give us $\frac{3 \alpha}{4} \left| B_t^\alpha\cap H(\alpha ,\opt) \right| \leq 4 \log (2n) \ \sum_{j=1}^n \left( \frac{3 \alpha}{2 \fac} \right) = \frac{6 \alpha n \log (2n)}{\fac}$. Simplifying we obtain the desired bound $|B_t^\alpha\cap H(\alpha ,\opt)|\leq \frac{8n\log(2n)}{\fac}$. 
\end{proof}

Recall that $L(\alpha, \opt) \coloneqq \{i\in[n] : v_i(\omega_i)<\alpha\}$ and $\widehat{L} (\alpha, \x) \coloneqq \left\{i\in[n] : v_i(x_i)<\frac{\alpha}{8 \fac} \right\}$.

\begin{lemma}
\label{common:lem3}
For any $\alpha  \in \left\{ \frac{1}{2^k} : 1 \leq k \leq \log(2n) \right\}$ we have $|\widehat{L} \left( \alpha, \x \right)| \leq |L(\alpha, \opt)| + \frac{8n \log(2n)}{\fac}$.
\end{lemma}
\begin{proof}
We begin by showing that $\widehat{L}(\alpha, \x) \subseteq B^\alpha_T$: consider any agent $i \in \widehat{L}(\alpha, \x)$; by definition of this set, $ v_i(x_i) <  \frac{\alpha}{8 \fac}$ and, hence, agent $i \in A_T^\alpha$. In addition, given that the value of the good in the last round is at most $\frac{1}{n^2} \leq \frac{\alpha}{4}$, we get $i\in B_T^\alpha$. Since this containment holds for each agent $i \in \widehat{L}(\alpha, \x)$, we obtain $\widehat{L}(\alpha, \x) \subseteq B_T^\alpha$. 

This containment and Lemma \ref{commonlem:1} lead to the stated bound: 
\begin{align*}
|\widehat{L}(\alpha, \x)| & = |\widehat{L}(\alpha, \x) \cap L(\alpha ,\opt)| + |\widehat{L}(\alpha, \x) \cap H(\alpha ,\opt)| \tag{$H(\alpha, \opt)$ and $L(\alpha, \opt)$ partition $[n]$}\\ 
& \leq |L(\alpha ,\opt)| + |\widehat{L}(\alpha, \x) \cap H(\alpha ,\opt)| \\ 
& \leq |L(\alpha ,\opt)| + |B_T^\alpha \cap H(\alpha ,\opt)| \tag{since $\widehat{L}(\alpha, \x) \subseteq B_T^\alpha$} \\
& \leq |L(\alpha ,\opt)| + \frac{8 n \log (2n)}{\fac} \tag{Lemma \ref{commonlem:1}} 
\end{align*}
This completes the proof. 
\end{proof}

Recall that $S_\fac (\opt) \coloneqq \left\{ i \in [n] : v_i(x_i) < \frac{1}{2\fac} v_i(\omega_i) \right\}$. The following lemma establishes that the number of such sub-optimal agents decreases linearly with the chosen parameter $\fac$. 

\begin{lemma}\label{commonlem:2}
The number of sub-optimal agents $\left|S_\fac (\opt) \right|\leq\frac{8n \log^2(2n)}{\fac}$.
\end{lemma}
\begin{proof}
We partition the set of agents $S_\fac (\opt)$ into $\log (2n)$ subsets, based on the values $v_i(\omega_i)$s. Specifically, for each $\alpha \in \left\{ \frac{1}{2^k} : 1 \leq k \leq \log (2n) \right\}$, write subset $S\left( \alpha \right) \coloneqq \left\{i \in S_\fac (\opt) : \alpha \leq v_i(\omega_i) < 2\alpha \right\}$. Since $\frac{1}{2n} \leq v_j(\omega_j) < 1$, for all agents $j \in [n]$, the subsets $S(\alpha)$s form a partition of $S_\fac(\opt)$. Therefore, $\sum_{k=1}^{\log (2n)} \left|S\left(\frac{1}{2^k}\right)\right|=|S_\fac (\opt)|$.  

Next, we will show that $S(\alpha) \subseteq B^\alpha_T \cap H(\alpha ,\opt)$ and apply Lemma \ref{commonlem:1}. Towards this, note that for each agent $i\in S(\alpha) \subseteq S_\fac (\opt)$, we have $v_i(x_i) < \frac{v_i(\omega_i)}{2 \fac} < \frac{2\alpha}{2\fac}=\frac{\alpha}{\fac}$. Since $\sum_{s=1}^T x^{s, \alpha}_i v^s_i \leq v_i(x_i)$, we get that agent $i$ continues to be in the active set (for $\alpha$) throughout the execution of the algorithm, i.e., $i \in A^\alpha_T$. In fact, agent $i \in B^\alpha_T$, since the value of the last good $T$ is at most $1/n^2$, which in turn in upper bounded by $\alpha/4$ (see Line \ref{line:btalpha} in \alg($\fac$)).  

Furthermore, for each agent $i \in S(\alpha)$, we have $v_i(\omega_i) \geq \alpha$, i.e., $ i \in H(\alpha ,\opt)$. These observations imply that every agent $ i \in S(\alpha)$ is contained in $ B^\alpha_T \cap H(\alpha ,\opt)$ as well; equivalently, $S(\alpha) \subseteq B^\alpha_T \cap H(\alpha ,\opt)$. Therefore, Lemma \ref{commonlem:1} gives us $\left|S\left( \alpha \right)\right|\leq \frac{8n\log(2n)}{\fac}$, for all $\alpha \in \left\{ \frac{1}{2^k} : 1 \leq k \leq \log (2n) \right\}$.  

This establishes the Lemma: $|S_\fac (\opt) | = \sum_{k=1}^{\log (2n)} \left| S \left( \frac{1}{2^k} \right) \right| \leq \log (2n) \ \frac{8n \log(2n)}{\fac}=\frac{8n \log^2(2n)}{\fac}$.  
\end{proof}
\section{Universal Online Algorithm for Maximizing $p$-Mean Welfare}
\label{section:universal}

Considering the $p$-mean welfare maximization problem simultaneously for all $p \leq 1$, this section establishes an online algorithm with a universal competitive ratio of $O(\sqrt{n}\log n)$. Specifically, we establish Theorem \ref{theorem:universal-comp-ratio} by executing our algorithm, $\alg(\fac)$, with threshold $\fac = 8\sqrt{n} \ \log (2n)$.

Throughout this section we will consider the allocation $\x=(x_i)_i$ returned by $\alg\left( 8\sqrt{n}\log (2n) \right)$ and establish an $O(\sqrt{n}\log n)$-competitive ratio for different values of $p$: Subsection \ref{section:universal-egalitarian} details the guarantee for egalitarian welfare ($p=-\infty$) and Subsection \ref{section:universal-nash} for Nash social welfare ($p=0$). Subsections \ref{section:universal-p-lessthanone} and \ref{section:universal-p-morethanone} address $p\leq -1$ and $p >-1$, respectively. Together, Subsections \ref{section:universal-egalitarian} to \ref{section:universal-p-morethanone} prove Theorem \ref{theorem:universal-comp-ratio}.

The following lemma holds for $\fac=8 \sqrt{n} \log (2n)$ and any arbitrary allocation $\opt=(\omega_i)_i$ with the property that $v_i(\omega_i) \geq \nicefrac{1}{2n}$, for all agents $i \in [n]$. We will invoke the lemma in the following subsections for different values of $p \leq 1$ and with $\opt$ as allocations that (approximately) maximize the $p$-mean welfare (see Remark \ref{remark:near-opt}). 

Recall that, complementary to $H(\alpha, \opt)$, the set $ L(\alpha, \opt) \coloneqq \{i\in[n] : v_i(\omega_i)<\alpha\}$. Also, with $\fac=8 \sqrt{n} \ \log (2n)$, the set  $\widehat{L} (\alpha, \x)$ contains the agents that have value less than $\frac{\alpha}{8 \fac} = \frac{\alpha}{64 \sqrt{n} \log (2n)}$ in allocation $\x$, i.e., 
$\widehat{L} (\alpha, \x) = \left\{i\in[n] : v_i(x_i)<\frac{\alpha}{64 \sqrt{n} \log (2n)} \right\}$. The following two lemmas build upon Lemma \ref{common:lem3} specifically for {$\fac= {8 \sqrt{n} \log (2n)}$}. 

\begin{lemma}\label{uni:lem2}
For parameter $\fac = 8 \sqrt{n} \log (2n)$ and any $\alpha  \in \left\{ \frac{1}{2^k} : 1 \leq k \leq \log(2n) \right\}$, if $|L(\alpha ,\opt)|\leq \sqrt{n}$, then $\widehat{L} (\alpha, \x)=\emptyset$.
\end{lemma}
\begin{proof}
Using the upper bound $|L(\alpha ,\opt)|\leq \sqrt{n}$, we will show that no agent $i \in [n]$ is contained in $\widehat{L} (\alpha, \x)$. Towards this, note that if, in $\alg(\fac)$, agent $i$ was removed directly from the active set $A_t^\alpha$ in some iteration $t$, then $v_i(x_i) \geq \frac{\alpha}{\fac} \geq \frac{\alpha}{64 \sqrt{n} \log (2n)}$ (Line \ref{line:atalpha}). In such a case, {$i \notin \widehat{L}(\alpha, \x)$}. Hence, for the rest of the proof we assume that $i \in A_t^\alpha$ for all iterations $1 \leq t \leq T$. Let $f$ be the iteration in which agent $i$ was included in $B_{f}^\alpha$ for the first time, i.e.,  $\sum_{s=f-1}^{T}v_i^{s}\geq \frac{\alpha}{4}$ and $\sum_{s=f}^{T}v_i^{s} < \frac{\alpha}{4}$ (Line \ref{line:btalpha}). Since all the goods have value at most $1/n^2$, we get 
\begin{align}
\sum_{s=f}^T v_i^{s} =\sum_{s=f-1}^T v_i^{s}-v_i^{f-1} \geq \frac{\alpha}{4}-\frac{1}{n^2} \geq \frac{\alpha}{8} \label{ineq:uni-suff-val}
\end{align}

Next, note that Lemma \ref{commonlem:1} gives us $|B_{s}^\alpha \cap H(\alpha ,\opt)| \leq \frac{8 n \log (2n)}{\fac} = \sqrt{n}$ for all $s \geq f$. Furthermore, we have $|L(\alpha ,\opt)|\leq \sqrt{n}$. Therefore, for all $s \geq f$, we can upper bound $B_{s}^\alpha$ as follows
\begin{align}
|B_{s}^\alpha| =|B_{s}^\alpha \cap H(\alpha ,\opt)| + |B_{s}^\alpha\cap L(\alpha ,\opt)| \leq \sqrt{n} + \sqrt{n} = 2 \sqrt{n} \label{ineq:card-bound}
\end{align}
Since $i \in B_{s}^\alpha$ for all $s \geq f$, agent $i$ receives at least $\frac{1}{4\log(2n)|B_{s}^\alpha|}$ fraction of the good in every iteration $s \geq f$ (Line \ref{line:btalpha}). Therefore, the value achieved by agent $i$ is at least
\begin{equation*}
\sum_{s =f}^T \frac{1}{4\log(2n) |B_{s}^\alpha|} \ v_i^{s} \underset{\text{(via (\ref{ineq:card-bound}))}}{\geq} \frac{1}{8\sqrt{n}\log(2n)} \sum_{s = f}^T v_i^{s} \underset{\text{(via (\ref{ineq:uni-suff-val}))}}{\geq} \frac{1}{8\sqrt{n}\log(2n)}  \ \frac{\alpha}{8} = \frac{\alpha}{64\sqrt{n}\log(2n)}.
\end{equation*}

Therefore, even in the current case $i\notin \widehat{L}(\alpha, \x)$. This, overall, establishes that no agent $i$ is contained in  $\widehat{L}(\alpha, \x)$ (i.e., $\widehat{L}(\alpha, \x) =\emptyset$). The lemma stands proved. 
\end{proof}

The following lemma multiplicatively bounds the number of low-valued agents in allocation $\x$ in terms of the number of low-valued agents in $\opt$.

\begin{lemma}\label{uni:lem3}
For parameter $\fac = 8 \sqrt{n} \log (2n)$ and any $\alpha  \in \left\{ \frac{1}{2^k} : 1 \leq k \leq \log(2n) \right\}$, we have $|\widehat{L}(\alpha, \x)| \leq 2 |L(\alpha ,\opt)|$.
\end{lemma}
\begin{proof}
Lemma \ref{uni:lem2} ensures that, if $L(\alpha ,\opt)\leq \sqrt{n}$, then $|\widehat{L}(\alpha, \x)|=0$. Hence, in such a case, the stated bound holds: $|\widehat{L}(\alpha, \x)| \leq 2|L(\alpha ,\opt)|$. The complementary case (i.e., $|L(\alpha ,\opt)| > \sqrt{n}$) is addressed by Lemma \ref{common:lem3}:
\begin{align*}
|\widehat{L}(\alpha, \x)| & \leq |L(\alpha ,\opt)| + \frac{8n \log (2n)}{\fac} \\
& = |L(\alpha ,\opt)| + \sqrt{n} \tag{here $\fac= 8 \sqrt{n} \cdot \log (2n)$} \\
& \leq 2|L(\alpha ,\opt)|. 
\end{align*}
This completes the proof. 
\end{proof}

\subsection{Universal Guarantee for Egalitarian Welfare ($\pmb{p=-\infty}$)}
\label{section:universal-egalitarian}
In this subsection we show that $\alg(\fac)$, with $\fac = 8 \sqrt{n} \cdot \log (2n)$, achieves a competitive ratio of $O\left(\sqrt{n} \cdot \log n\right)$ for maximizing egalitarian welfare ${\rm M}_{-\infty}(\cdot)$.\footnote{As mentioned previously, this competitive ratio is tight, up to a log factor.} Here, let $\opt=(\omega_i)_i$ denote an allocation with egalitarian welfare, ${\rm M}_{-\infty}(\opt)$, at least half of the optimal egalitarian welfare and $v_i(\omega_i) \geq \nicefrac{1}{2n}$, for all $i \in [n]$; such an allocation is guaranteed to exist (Remark \ref{remark:near-opt}). 

We will show that the allocation $\x=(x_i)_i$---computed by $\alg \left( 8 \sqrt{n} \log (2n) \right)$---satisfies ${\rm M}_{-\infty} (\x) \geq \frac{1}{128 \sqrt{n} \log (2n)} {\rm M}_{-\infty} (\opt)$ and, hence, obtain the stated competitive ratio for egalitarian welfare.  

Write $\kappa$ to denote the integer that satisfies $\frac{1}{2^\kappa} \leq {\rm M}_{-\infty} (\opt) < \frac{2}{2^{\kappa}}$. Since $v_i(\omega_i) \geq \nicefrac{1}{2n}$, for all agents $i \in [n]$, we have ${\rm M}_{-\infty}(\opt) = \min_i v_i(\omega_i) \geq \nicefrac{1}{2n}$. Hence, $\kappa \geq \log (2n)$. Setting $\tilde{\alpha} \coloneqq \nicefrac{1}{2^\kappa}$, we invoke Lemma \ref{uni:lem2} and note that $L(\tilde{\alpha} ,\opt)=\emptyset$ and, hence, $\widehat{L}(\tilde{\alpha}, \x)=\emptyset$. Therefore, all agents $i \in [n]$ satisfy $v_i(x_i) \geq \frac{\tilde{\alpha}}{64\sqrt{n}\log(2n)} > \frac{{\rm M}_{-\infty}(\opt)}{128 \sqrt{n}\log(2n)}$. Equivalently, ${\rm M}_{-\infty} (\x) \geq \frac{1}{128 \sqrt{n} \log (2n)} {\rm M}_{-\infty} (\opt)$ and the stated competitive ratio holds. 

The next subsection shows that the allocation $\x$ obtains an analogous competitive guarantee for Nash social welfare as well. 

\subsection{Universal Guarantee for Nash Social Welfare ($\pmb{p=0}$)}
\label{section:universal-nash}
In this subsection, let $\opt=(\omega_i)_i$ denote an allocation with Nash social welfare, ${\rm M}_{0}(\opt)$, at least half of the optimal Nash social welfare and $v_i(\omega_i) \geq \nicefrac{1}{2n}$, for all $i \in [n]$. Recall that $S_\fac(\opt)$ is a subset of agents that are $(2\fac)$-sub-optimal in $\x$; specifically, $S_\fac (\opt) \coloneqq \left\{ i \in [n] : v_i(x_i) < \frac{1}{2\fac} v_i(\omega_i) \right\}$. In the current context we have $\fac = 8 \sqrt{n} \cdot \log (2n)$ and, hence, Lemma \ref{commonlem:2}, gives us $|S_\fac(\opt)| \leq \frac{8n \log^2(2n)}{\fac}=\sqrt{n} \cdot \log(2n)$. Using this upper bound on $|S_\fac(\opt)|$ we establish the competitive ratio for Nash social welfare:
\begin{align*}
\left( \frac{\prod_{i=1}^n v_i(\omega_i)}{\prod_{i=1}^n v_i(x_i)}\right)^{\frac{1}{n}} &=  \left( \prod_{i=1}^n\frac{v_i(\omega_i)}{v_i(x_i)}\right)^{\frac{1}{n}} = \left(\prod_{i\in S_\fac(\opt)}\frac{v_i(\omega_i)}{v_i(x_i)}\right)^{\frac{1}{n}} \left(\prod_{i\in [n]\setminus S_\fac(\opt)}\frac{v_i(\omega_i)}{v_i(x_i)}\right)^{\frac{1}{n}}\\ 
& \leq (2n)^{\frac{|S_\fac(\opt)|}{n}} \left(\prod_{i\in [n]\setminus S_\fac(\opt)}\frac{v_i(\omega_i)}{v_i(x_i)}\right)^{\frac{1}{n}} \tag{$v_i(\omega_i) \leq 1$ \& $v_i(x_i) \geq \frac{1}{2n}$, for all $i$; Line \ref{line:prop}}\\
& \leq (2n)^{\frac{|S_\fac(\opt)|}{n}} \left(16\sqrt{n}\log(2n)\right)^{\frac{n-|S_\fac(\opt)|}{n}} \tag{$v_i(x_i)\geq\frac{v_i(\omega_i)}{16\sqrt{n}\log(2n)}$ for all $i \notin S_\fac(\opt)$}\\
& \leq (2n)^{\frac{\sqrt{n}\cdot \log(2n)}{n}}\left(16\sqrt{n}\log(2n)\right) \tag{since $|S_\fac(\opt)|\leq \sqrt{n}\log(2n)$}\\
& = 2^{\log (2n) \cdot \frac{\log(2n)}{\sqrt{n}}}\left(16\sqrt{n}\log(2n)\right) \\ & \leq 32\sqrt{n}\log(2n).
\end{align*}
Hence, the allocation $\x$ is $O\left(\sqrt{n}\log n \right)$-competitive for Nash Social Welfare.

\subsection{Universal Guarantee for {\pmb{$p \leq -1$}}}
\label{section:universal-p-lessthanone}
This subsection shows that the allocation $\x = (x_i)_{i}$---computed by $\alg \left(8 \sqrt{n} \cdot \log (2n) \right)$---achieves an $O(\sqrt{n}\log n)$-approximation for $p$-mean welfare maximization, for all $p\leq -1$. 

For this subsection, fix any $p \leq -1$ and write $\opt=(\omega_i)_i$ to denote an allocation with $p$-mean welfare, ${\rm M}_{p}(\opt)$, at least half of the optimal $p$-mean welfare and $v_i(\omega_i) \geq \nicefrac{1}{2n}$, for all $i \in [n]$. 

Let $D(\alpha,\opt):=\{i\in[n]:\alpha/2\leq v_i(\omega_i) < \alpha\}$ denote the set of agents which achieve value at least $\alpha/2$ and less than $\alpha$ in the allocation $\opt$. Note that $D(\alpha, \opt) = L(\alpha, \opt) \setminus L\left({\alpha}/{2}, \opt \right)$; in fact, by the definitions of these sets we have
\begin{align}
    |D(\alpha, \opt)| = |L(\alpha, \opt)| - |L\left({\alpha}/{2}, \opt \right)| \label{eq:dll}
\end{align}

Also, the bounds $\nicefrac{1}{2n} \leq v_i(\omega_i) < 1$, for all $i$, ensure that the sets $\left\{ D\left( \frac{1}{2^k}, \opt \right) \right\}_{k=0}^{\log (2n)}$ partition all of $[n]$.

To establish the stated ratio between ${\rm M}_p(\x) = \left( \frac{1}{n} \sum_i v_i(x_i)^p \right)^{1/p}$ and ${\rm M}_p(\opt) = \left( \frac{1}{n} \sum_i v_i(\omega_i)^p \right)^{1/p}$, we first lower bound $\frac{1}{n} \sum_{i=1}^n v_i(\omega_i)^p$:
\begin{align}
\frac{1}{n} \sum_{i=1}^n v_i(\omega_i)^p &=\frac{1}{n} \ \sum_{k=0}^{\log (2n)}\sum_{i\in D\left(\frac{1}{2^k},\opt\right)}(v_i(\omega_i))^p =\frac{1}{n} \sum_{k=0}^{\log (2n)}\sum_{i\in D\left(\frac{1}{2^k},\opt\right)}\left(\frac{1}{v_i(\omega_i)}\right)^{|p|} \tag{since $p<0$} \nonumber\\
&> \frac{1}{n}  \sum_{k=0}^{\log (2n)} \sum_{i\in D\left(\frac{1}{2^k},\opt\right)}2^{k|p|} \tag{$v_i(\omega_i) < 1/2^k$ for all $i \in D\left(\frac{1}{2^k},\opt\right)$} \nonumber \\
&=\frac{1}{n} \sum_{k=0}^{\log (2n)}\left|D\left(\frac{1}{2^k},\opt\right)\right|2^{k|p|} \nonumber \\
&=\frac{1}{n}  \sum_{k=0}^{\log (2n)}\left(\left|L\left(\frac{1}{2^k},\opt\right)\right|-\left|L\left(\frac{1}{2^{k+1}},\opt\right)\right|\right)2^{k|p|} \tag{via inequality (\ref{eq:dll})} \nonumber \\
& = \frac{1}{n} \left( \sum_{k=1}^{\log (2n)}\left|L\left(\frac{1}{2^k},\opt\right)\right|\left(2^{k|p|}-2^{(k-1)|p|}\right)+\left|L\left(1,\opt\right)\right|\right)\label{daw:st5}
\end{align}
For the last equality we use the fact that $L\left(\frac{1}{4n},\opt\right)=\emptyset$.

We next obtain a complementary bound considering allocation $\x=(x_i)_i$. Towards this, write $\widehat{D}(\alpha, \x) \coloneqq \left\{i\in[n] : \frac{\alpha}{128\sqrt{n}\log(2n)}\leq v_i(x_i) < \frac{\alpha}{64\sqrt{n}\log(2n)} \right\}$ and let $\widehat{U} \coloneqq \left\{i\in[n] : v_i(x_i)< \frac{1}{64\sqrt{n}\log(2n)}\right\}$. Recall that $v_i(x_i) \geq \frac{1}{2n}$ for all agents $i$ (Line \ref{line:prop}) and, hence, the sets $\left\{\widehat{D}\left(\frac{1}{2^k}, \x \right) \right\}_{k=0}^{\log (2n)}$ form a partition of $\widehat{U}$. Towards bounding $\frac{1}{n} \sum_{i=1}^n v_i(x_i)^p =  \frac{1}{n} \sum_{i\in \widehat{U}} \ v_i(x_i)^p + \frac{1}{n} \sum_{i\in [n] \setminus \widehat{U}} \ v_i(x_i)^p$, we address the two summands separately. 

\begingroup
\allowdisplaybreaks
\begin{align}
\frac{1}{n} \sum_{i \in \widehat{U}} v_i(x_i)^p &=\frac{1}{n}  \sum_{k=0}^{\log (2n)}\sum_{i\in \widehat{D}\left(\frac{1}{2^k}, \x \right)} v_i(x_i)^p \nonumber \\
&=\frac{1}{n} \sum_{k=0}^{\log (2n)}\sum_{i\in \widehat{D}\left(\frac{1}{2^k}, \x \right)}\left(\frac{1}{v_i(x_i)}\right)^{|p|} \tag{since $p<0$} \\
&\leq \frac{1}{n} \sum_{k=0}^{\log (2n)}\sum_{i\in \widehat{D}\left(\frac{1}{2^{k}}, \x \right)} 2^{k|p|}(128\sqrt{n}\log (2n))^{|p|} \tag{considering $i \in \widehat{D}\left(\frac{1}{2^{k}}\right)$ and $|p| \ge 1$} \nonumber \\
& = (128\sqrt{n}\log (2n))^{|p|} \cdot \frac{1}{n} \sum_{k=0}^{\log (2n)}\left|\widehat{D}\left(\frac{1}{2^{k}}, \x \right)\right| \ 2^{k|p|} \label{ineq:dhat} 
\end{align}
\endgroup

Since $\left|\widehat{D}\left(\frac{1}{2^{k}}, \x \right)\right| = \left|\widehat{L}\left(\frac{1}{2^k}, \x \right)\right|-\left|\widehat{L}\left(\frac{1}{2^{k+1}}, \x \right)\right|$, for all $0 \leq k \leq \log (2n)$, inequality (\ref{ineq:dhat}) reduces to 
\begin{align}
\frac{1}{n} \sum_{i\in \widehat{U}} v_i(x_i)^p & =(128\sqrt{n}\log (2n))^{|p|} \ \frac{1}{n} \sum_{k=0}^{\log (2n)}\left(\left|\widehat{L}\left(\frac{1}{2^k}, \x \right)\right|-\left|\widehat{L}\left(\frac{1}{2^{k+1}}, \x \right)\right|\right)2^{k|p|} \nonumber \\
&=(128\sqrt{n}\log (2n))^{|p|} \ \frac{1}{n} \left( \sum_{k=1}^{\log (2n)}\left|\widehat{L}\left(\frac{1}{2^k}, \x \right)\right|\left(2^{k|p|}-2^{(k-1)|p|}\right)+\left|\widehat{L}\left(1, \x \right)\right|\right) \tag{since $|\widehat{L}\left(\frac{1}{4n}, \x\right)|=0$} \\
& \leq (128\sqrt{n}\log (2n))^{|p|} \  \frac{1}{n} \cdot 2 \left( \sum_{k=1}^{\log (2n)}\left|L\left(\frac{1}{2^k},\opt \right)\right|\left(2^{k|p|}-2^{(k-1)|p|}\right)+\left|L\left(1,\opt\right)\right|\right) \tag{via Lemma \ref{uni:lem3}} \\
& \leq 2(128\sqrt{n}\log (2n))^{|p|} \ \  \frac{1}{n} \sum_{i=1}^n v_i(w_i)^p \tag{via inequality (\ref{daw:st5})} \\
& \leq (256\sqrt{n}\log (2n))^{|p|} \ \ \frac{1}{n} \sum_{i=1}^n v_i(w_i)^p \tag{since $|p|\geq 1$} 
\end{align}

Now, we upper bound $\frac{1}{n} \sum_{i \notin \widehat{U}} v_i(x_i)^p$:
\begin{align*}
\frac{1}{n} \sum_{i \notin \widehat{U}} v_i(x_i)^p & \leq \frac{1}{n} \sum_{i\notin \widehat{U}} \left(\frac{1}{64\sqrt{n}\log(2n)} \right)^p  \tag{$v_i(x_i)\geq \frac{1}{64\sqrt{n}\log(2n)}$, for all $i\notin \widehat{U}$ \& $p\leq -1$}\\
& = \left(64\sqrt{n}\log(2n) \right)^{|p|} \ \frac{1}{n} \sum_{i\notin \widehat{U}} 1  \\
& \leq \left( 64\sqrt{n}\log(2n) \right)^{|p|}\ \frac{1}{n} \sum_{i\notin \widehat{U}} v_i(\omega_i)^p  \tag{since $v_i(\omega_i) \leq 1$, for all $i$, and $p \leq -1$}\\
&\leq \left( 64\sqrt{n}\log(2n)\right)^{|p|} \frac{1}{n} \sum_{i =1}^{n} \ v_i(\omega_i)^p
\end{align*}
Using these upper bounds we establish the competitive ratio:
\begin{align}
\frac{1}{n} \sum_{i=1}^{n} v_i(x_i)^p &=\frac{1}{n} \sum_{i\in \widehat{U}} v_i(x_i)^p +\frac{1}{n} \sum_{i\notin \widehat{U}} v_i(x_i)^p \nonumber \\
&\leq (256\sqrt{n}\log (2n))^{|p|} \ \frac{1}{n} \sum_{i=1}^n {v_i(\omega_i)^p} + (64\sqrt{n}\log(2n))^{|p|} \ \frac{1}{n} \sum_{i=1}^n  v_i(\omega_i)^p \nonumber \\
&\leq (320\sqrt{n}\log(2n))^{|p|}\ \frac{1}{n} \sum_{i =1}^n v_i(\omega_i)^p \label{ineq:UUhat}  
\end{align}
The last inequality follows from the fact that $|p| \geq 1$. Also, here we have $p \leq -1$ and, hence, exponentiating both sides of equation (\ref{ineq:UUhat}) by $1/p$ we get $\left( \frac{1}{n} \sum_i v_i(x_i)^p \right)^{1/p} \geq \frac{1}{320 \sqrt{n} \log (2n)} \left(\frac{1}{n} \sum_i v_i(\omega_i)^p \right)^{1/p}$, i.e., ${\rm M}_p(\x)\geq \frac{1}{320\sqrt{n}\log(2n)} {\rm M}_p(\opt)$. This, overall, establishes the stated competitive ratio of allocation $\x$ for all $p\leq -1$.

\subsection{Universal Guarantee for {\pmb{$p > -1$}}}
\label{section:universal-p-morethanone}
This subsection shows that the allocation $\x = (x_i)_{i}$---computed by $\alg \left(8 \sqrt{n} \cdot \log (2n) \right)$---achieves an $O(\sqrt{n}\log n)$-approximation for $p$-mean welfare maximization, for all nonzero $p > -1$.\footnote{The case of $p = 0$ is addressed in Subsection \ref{section:universal-nash}.} 

For this subsection, fix any nonzero $p > -1$ and write $\opt=(\omega_i)_i$ to denote an allocation with $p$-mean welfare, ${\rm M}_{p}(\opt)$, at least half of the optimal $p$-mean welfare and $v_i(\omega_i) \geq \nicefrac{1}{2n}$, for all $i \in [n]$.

Recall that $S_\fac(\opt)$ denotes the set of $(2\fac)$-sub-optimal agents in $\x$. In the current setting, $\fac=8 \sqrt{n} \cdot \log (2n)$ and, hence, via 
Lemma \ref{commonlem:2} we get $|S_\fac(\opt)| \leq \frac{8n \log^2(2n)}{\fac} = \sqrt{n}\log(2n)$. For notational convenience, throughout this subsection we will use $S$ for the set $S_\fac(\opt)$ and write $S^c \coloneqq [n]\setminus S$. Note that 
\begin{align}
|S| \leq \sqrt{n}\log(2n) \label{ineq:sizeS}    
\end{align}
Furthermore, let  $\W_{S}$ and $\W_{S^c}$ denote the $p$-mean welfare of the (sub)allocations $(\omega_i)_{i\in S}$ and $(\omega_i)_{i\in S^c}$, respectively, $ \W_{S} \coloneqq \left(\frac{1}{|S|} \sum_{i\in S} v_i(\omega_i)^p \right)^{1/p}$ and $ \W_{S^c} \coloneqq \left(\frac{1}{|S^c|} \sum_{i\in S^c} v_i(\omega_i)^p \right)^{1/p}$. Towards establishing the competitive ratio we first lower bound $\W_{S^c}$.
\begin{lemma}\label{ulem:4}
For any nonzero $p \in (-1, 1]$, we have $\W_{S^c} \geq {\rm M}_p (\opt) -\frac{\log(2n)}{\sqrt{n}}$.
\end{lemma}
\begin{proof}
We can express the $p$-mean welfare of allocation $\opt$ as
\begin{align*}
{\rm M}_p (\opt) & =\left(\frac{|S^c|}{n}\left(\frac{1}{|S^c|} \ \sum_{i\in S^c} v_i(\omega_i)^p\right) + \frac{|S|}{n} \left(\frac{1}{|S|} \ \sum_{i\in S}  v_i(\omega_i)^p \right)\right)^{\frac{1}{p}} \\
&=\left(\frac{|S^c|}{n}\left( \W_{S^c} \right)^p + \frac{|S|}{n} \left( \W_{S} \right)^p \right)^{1/p}\\
& \leq \frac{|S^c|}{n} \left( \W_{S^c}\right) + \frac{|S|}{n} \left( \W_S \right) \tag{via the generalized mean inequality and $p\leq 1$}\\
& \leq \frac{|S^c|}{n} \W_{S^c} + \frac{\sqrt{n} \  \log(2n)}{n} \cdot 1 \tag{via inequality (\ref{ineq:sizeS}) and $\W_{S}\leq 1$} \\
& \leq \W_{S^c} + \frac{\log(2n)}{\sqrt{n}}
\end{align*}
Therefore, the lemma follows: $\W_{S^c} \geq {\rm M}_p (\opt) -\frac{\log(2n)}{\sqrt{n}}$.
\end{proof}

Now, we establish the competitive ratio for the entire allocation $\x=(x_i)_{i \in [n]}$ by considering the following two (complementary and exhaustive) cases.

\noindent
\emph{Case {\rm I}:} ${\rm M}_p(\opt) \leq \frac{2\log(2n)}{\sqrt{n}}$. For this case, note that $v_i(x_i)\geq \frac{1}{2n}$, for all agents $i \in [n]$ (Line \ref{line:prop}). Therefore, 
\begin{align*}
{\rm M}_p (\x) & = \left(\frac{1}{n} \sum_{i=1}^n v_i(x_i)^p\right)^{1/p} \geq \frac{1}{2n} \geq \frac{1}{4\sqrt{n}\log(2n)} {\rm M}_p(\opt) \tag{since ${\rm M}_p(\opt) \leq \frac{2\log(2n)}{\sqrt{n}}$}
\end{align*}
Hence, in this case the stated competitive ratio holds. 

\noindent
\emph{Case {\rm II}:} ${\rm M}_p(\opt) > \frac{2\log(2n)}{\sqrt{n}}$. For this case, recall that, by definition, agents in $S^c$  are \emph{not} $2\fac = 16 \sqrt{n} \cdot \log (2n)$ sub-optimal, i.e., for every agent $i \in S^c$ the value $v_i(x_i) \geq \frac{1}{16\sqrt{n}\log(2n)} v_i(\omega_i)$. Therefore,  

\begingroup
\allowdisplaybreaks
\begin{align*}
{\rm M}_p (\x) & = \left(\frac{1}{n} \sum_{i=1}^n v_i(x_i)^p\right)^{1/p} \\
& = \left(\frac{1}{n} \sum_{i \in S^c}  v_i(x_i)^p + \frac{1}{n} \sum_{i \in S}  v_i(x_i)^p \right)^{1/p} \\
& \geq \left(\frac{1}{n} \sum_{i \in S^c}  \left(\frac{v_i(\omega_i)}{16\sqrt{n}\log(2n)} \right)^p + \frac{1}{n} \sum_{i \in S}  v_i(x_i)^p \right)^{1/p} \tag{monotonicity of $p$-mean} \\
& = \left( \frac{1}{n} \cdot \frac{1}{\left(16 \sqrt{n} \log (2n) \right)^p} \sum_{i \in S^c} v_i(\omega_i)^p + \frac{1}{n} \sum_{i \in S}  v_i(x_i)^p \right)^{1/p} \\
& = \left(\frac{1}{\left(16 \sqrt{n} \log (2n) \right)^p} \cdot \frac{|S^c|}{n} \left( \frac{1}{|S^c|} \sum_{i \in S^c} v_i(\omega_i)^p \right) + \frac{1}{n} \sum_{i \in S}  v_i(x_i)^p \right)^{1/p} \\
& = \left(\frac{1}{\left(16 \sqrt{n} \log (2n) \right)^p} \cdot \frac{|S^c|}{n} \left( \W_{S^c} \right)^p + \frac{1}{n} \sum_{i \in S}  v_i(x_i)^p \right)^{1/p} \\
& = \left( \frac{|S^c|}{n} \left( \frac{\W_{S^c}}{16 \sqrt{n} \log (2n)} \right)^p + \frac{1}{n} \sum_{i \in S}  v_i(x_i)^p \right)^{1/p}
\end{align*}
\endgroup

Using the fact that $v_i(x_i) \geq \nicefrac{1}{2n}$, for all agents $i$, the previous inequality reduces to 
\begin{align}
{\rm M}_p (\x) & \geq \left( \frac{|S^c|}{n} \left( \frac{\W_{S^c}}{16 \sqrt{n} \log (2n)} \right)^p + \frac{1}{n} \sum_{i \in S}  \left(\frac{1}{2n} \right)^p \right)^{1/p} \tag{monotonicity of $p$-mean} \\
& = \left( \frac{|S^c|}{n} \left( \frac{\W_{S^c}}{16 \sqrt{n} \log (2n)} \right)^p + \frac{|S|}{n} \left(\frac{1}{2n} \right)^p \right)^{1/p} \nonumber \\
& \geq \left( \frac{|S^c|}{n} \left( \frac{\W_{S^c}}{16 \sqrt{n} \log (2n)} \right)^{-1} + \frac{|S|}{n} \left(\frac{1}{2n} \right)^{-1} \right)^{-1} \label{ineq:harmonic}
\end{align}
The last equation follows from the generalized mean inequality and $p\geq -1$ (i.e., the $p$-mean, for all $p\geq -1$, is at least the harmonic mean).    
Simplifying equation (\ref{ineq:harmonic}) we obtain 
\begin{align*}
{\rm M}_p (\x) & \geq \frac{1}{\frac{|S^c|}{n}\ \frac{16\sqrt{n}\log(2n)}{ \W_{S^c} }+ 2 |S|}\\
& = \frac{\W_{S^c}}{\frac{|S^c|}{n}\ 16\sqrt{n}\log(2n) + 2 |S| \ \W_{S^c}}\\
& \geq \frac{\W_{S^c}}{16\sqrt{n}\log(2n) + 2\W_{S^c} \ \sqrt{n}\log(2n)} \tag{since $|S|\leq \sqrt{n}\log(2n)$; inequality (\ref{ineq:sizeS})} \\
& \geq \frac{\W_{S^c}}{16\sqrt{n}\log(2n) + 2 \sqrt{n}\log(2n)} \tag{since $\W_{S^c} \leq1$} \\
& \geq \frac{{\rm M}_p(\opt) -\nicefrac{\log(2n)}{\sqrt{n}}}{18\sqrt{n}\log(2n)} \tag{via Lemma \ref{ulem:4}} \\
& \geq \frac{{\rm M}_p(\opt)}{36 \sqrt{n}\log(2n)}   \tag{since ${\rm M}_p(\opt) > \frac{2\log(2n)}{\sqrt{n}}$}  
\end{align*}

This, overall, establishes that allocation $\x$ is $O\left( \sqrt{n} \cdot \log n \right)$-competitive for all nonzero $p \in (-1, 1]$.
\section{Tight Online Algorithms for Maximizing $\pmb{p}$-Mean Welfare}
\label{section:tight}

In contrast to Section \ref{section:universal}, where we worked with a single threshold $\fac$ ($=8 \sqrt{n} \cdot \log (2n)$), in this section we execute $\alg(\cdot)$ with $p$-specific thresholds and obtain essentially-tight competitive ratios for a wide range of the exponent parameter $p \leq 1$. Subsections \ref{section:tight-egalitarian} to \ref{sec:5.3} establish the upper bounds stated in Table \ref{table:UBLB}. The lower bounds in the table (proved in Section \ref{section:lower-bounds}) show that the upper bounds obtained here are tight---up to poly-log factors---for a multiple ranges of $p$. 

\subsection{Tight Guarantee for Egalitarian Welfare ($\pmb{p = -\infty}$)}
\label{section:tight-egalitarian}
As mentioned previously, a lower bound of $\sqrt{n}$ holds on the competitive ratio of any online algorithm for maximizing egalitarian welfare (see also Corollary \ref{corollary:lower-bound-egalitarian}). Therefore, the universally competitive algorithm (i.e., $\alg(\fac)$ with $\fac = 8 \sqrt{n} \log (2n)$) provides a tight---up to a log factor---competitive ratio (of $O\left(\sqrt{n} \log n \right)$) for egalitarian welfare.

\subsection{Tight Guarantee for Nash Social Welfare}
\label{sec:5.2}
This subsection shows that $\alg(\fac)$, with $\fac = 8\log^3 (2n)$, achieves a competitive ratio of $O\left(\log^3 n\right)$ for maximizing Nash social welfare ${\rm M}_{0}(\cdot)$.\footnote{This competitive ratio is tight, up to a poly-log factor.} 

Here, let $\opt=(\omega_i)_i$ denote an allocation with Nash social welfare, ${\rm M}_{0}(\opt)$, at least half of the optimal Nash social welfare and $v_i(\omega_i) \geq \nicefrac{1}{2n}$, for all $i \in [n]$ (see Remark \ref{remark:near-opt}). We will show that the allocation $\x=(x_i)_i$---computed by $\alg \left( 8\log^3 (2n) \right)$---satisfies ${\rm M}_{0} (\x) \geq \frac{1}{32\log^3 (2n)} {\rm M}_{0} (\opt)$ and, hence, achieves the stated competitive ratio for Nash social welfare.  

Recall that $S_\fac(\opt)$ denotes the set of $(2\fac)$-sub-optimal agents. Since $\fac =  8\log^3 (2n)$ here, we have $S_\fac (\opt) = \left\{ i \in [n] :  v_i(x_i) < \frac{1}{16 \log^3 n} v_i(\omega_i) \right\}$, and Lemma \ref{commonlem:2} gives us $|S_\fac(\opt)| \leq \frac{8n \log^2(2n)}{\fac}=\frac{n}{\log(2n)}$. Using this upper bound on $|S_\fac(\opt)|$ we establish the competitive ratio for Nash social welfare
\begin{align*}
\left( \frac{\prod_{i=1}^n v_i(\omega_i)}{\prod_{i=1}^n v_i(x_i)}\right)^{\frac{1}{n}} &=  \left( \prod_{i=1}^n\frac{v_i(\omega_i)}{v_i(x_i)}\right)^{\frac{1}{n}} = \left(\prod_{i\in S_\fac(\opt)}\frac{v_i(\omega_i)}{v_i(x_i)}\right)^{\frac{1}{n}} \left(\prod_{i\in [n]\setminus S_\fac(\opt)}\frac{v_i(\omega_i)}{v_i(x_i)}\right)^{\frac{1}{n}}\\ 
& \leq (2n)^{\frac{|S_\fac(\opt)|}{n}} \left(\prod_{i\in [n]\setminus S_\fac(\opt)}\frac{v_i(\omega_i)}{v_i(x_i)}\right)^{\frac{1}{n}} \tag{$v_i(\omega_i) <1$ and  $v_i(x_i) \geq \frac{1}{2n}$, for all $i$; Line \ref{line:prop}}\\
& \leq (2n)^{\frac{|S_\fac(\opt)|}{n}} \left(16\log^3(2n)\right)^{\frac{n-|S_\fac(\opt)|}{n}} \tag{$v_i(x_i)\geq\frac{v_i(\omega_i)}{16\log^3(2n)}$ for all $i \notin S_\fac(\opt)$}\\
& \leq (2n)^{\frac{n}{n\log(2n)}}\left(16\log^3(2n)\right) \tag{since $|S_\fac(\opt)|\leq \frac{n}{\log(2n)}$}\\
& = 2^{\log (2n) \cdot \frac{1}{\log(2n)}}\left(16\log^3(2n)\right) \\ & = 32\log^3(2n)
\end{align*}
Hence the allocation $\x$ is $O\left(\log^3n \right)$-competitive for Nash social welfare.

\subsection{Improved  Guarantee for {\pmb{$p\leq -1$}}}
\label{sec:5.7}
We will show in Section \ref{section:lower-bounds} that, for maximizing $p$-mean welfare with parameter $p \leq -1$, any online algorithm incurs a competitive ratio of $O \left( n^{\frac{1}{2 + (1/|p|) } - \varepsilon} \right)$, for any constant $\varepsilon>0$. Therefore, the universal competitive guarantee of $O\left( \sqrt{n} \log n \right)$ (obtained in Section \ref{section:universal-p-lessthanone}) is commensurate with the lower bound  when $p \leq -1$.   

\subsection{Improved Guarantee for $\pmb{-1\leq p\leq -1/4}$}
\label{section:tight-minus-four-one}

In this subsection we show that $\alg(\fac)$, with $\fac = 8n^{\frac{|p|}{|p|+1}}\log^2(2n)$, achieves a competitive ratio of $O\left(n^{\frac{|p|}{|p|+1}}\log^2n\right)$ for maximizing $p$-mean welfare with $-1 \leq p \leq -\frac{1}{4}$. Here, let $\opt=(\omega_i)_i$ denote an allocation with $p$-mean welfare, ${\rm M}_{p}(\opt)$, at least half of the optimal $p$-mean welfare and $v_i(\omega_i) \geq \nicefrac{1}{2n}$, for all $i \in [n]$; such an allocation is guaranteed to exist (Remark \ref{remark:near-opt}).

We will show that the allocation $\x=(x_i)_i$---computed by $\alg \left(  8n^{\frac{|p|}{|p|+1}}\log^2(2n) \right)$---satisfies ${\rm M}_{p} (\x) \geq \frac{1}{2^{1+1/|p|}\fac} \  {\rm M}_{p} (\opt)$ and, hence, obtain the stated competitive ratio for $p$-mean welfare.  

Recall that $S_\fac(\opt)$ denotes the set of $(2\fac)$-sub-optimal agents in $\x$. For notational convenience, throughout this subsection we will use $S$ for the set $S_\fac(\opt)$ and write $S^c \coloneqq [n]\setminus S$. In the current setting, $\fac=8n^{\frac{|p|}{|p|+1}}\log^2(2n)$ and, hence, via 
Lemma \ref{commonlem:2} we get $|S| = |S_\fac(\opt)| \leq \frac{8n \log^2(2n)}{\fac} =n^{\frac{1}{|p|+1}}$.

Using the facts that $v_i(x_i)\geq \frac{1}{2n}$ and $v_i(\omega_i)<1$, for all $i\in[n]$,  along with $v_i(x_i)\geq \frac{v_i(\omega_i)}{2\fac}$ for all $i\in S^c$, we upper bound $\frac{{\rm M}_p(\opt)}{{\rm M}_p(\x)}$ as follows:

\begingroup
\allowdisplaybreaks
\begin{align*}
\left(\frac{\frac{1}{n} \sum_{i=1}^n v_i(\omega_i)^p}{\frac{1}{n}\sum_{i=1}^n v_i(x_i)^p}\right)^{1/p}&= \left(\frac{\sum_{i=1}^n \left(\frac{1}{v_i(x_i)}\right)^{|p|}}{\sum_{i=1}^n \left(\frac{1}{v_i(\omega_i)}\right)^{|p|}}\right)^{1/|p|} \tag{since $p<0$}\\
& = \left(\frac{\sum_{i\in S}\left(\frac{1}{v_i(x_i)}\right)^{|p|}+\sum_{i\in S^c}\left(\frac{1}{v_i(x_i)}\right)^{|p|}}{\sum_{i\in S}\left(\frac{1}{v_i(\omega_i)}\right)^{|p|}+\sum_{i\in S^c}\left(\frac{1}{v_i(\omega_i)}\right)^{|p|}}\right)^{1/|p|}\\
& \leq \left(\frac{\sum_{i\in S}\left(\frac{1}{v_i(x_i)}\right)^{|p|}+(2\fac)^{|p|}\sum_{i\in S^c}\left(\frac{1}{v_i(\omega_i)}\right)^{|p|}}{\sum_{i\in S}\left(\frac{1}{v_i(\omega_i)}\right)^{|p|}+\sum_{i\in S^c}\left(\frac{1}{v_i(\omega_i)}\right)^{|p|}}\right)^{1/|p|} \tag{$v_i(x_i) \geq \frac{1}{2 \fac} v_i(\omega_i)$, for all $i \in S^c$, and $p<0$} \\
&  \leq \left(\frac{\sum_{i\in S}(2n)^{|p|}+(2\fac)^{|p|}\sum_{i\in S^c}\left(\frac{1}{v_i(\omega_i)}\right)^{|p|}}{\sum_{i\in S}\left(\frac{1}{v_i(\omega_i)}\right)^{|p|}+\sum_{i\in S^c}\left(\frac{1}{v_i(\omega_i)}\right)^{|p|}}\right)^{1/|p|}\tag{$v_i(x_i) \geq \frac{1}{2n}$, for all $i$, and $p <0$}\\
&  \leq \left(\frac{(2n)^{|p|}|S|+(2\fac)^{|p|}\sum_{i\in S^c}\left(\frac{1}{v_i(\omega_i)}\right)^{|p|}}{|S|+\sum_{i\in S^c}\left(\frac{1}{v_i(\omega_i)}\right)^{|p|}}\right)^{1/|p|} \tag{since $v_i(\omega_i) \leq 1$ for all $i$} \\
&  = 2n \left(\frac{ |S|+\left(\frac{\fac}{n}\right)^{|p|} \sum_{i\in S^c}\left(\frac{1}{v_i(\omega_i)}\right)^{|p|}}{|S|+\sum_{i\in S^c}\left(\frac{1}{v_i(\omega_i)}\right)^{|p|}}\right)^{1/|p|}
\end{align*}
\endgroup
Now, the fact that the function $f(x) \coloneqq \frac{x+a}{x+b}$---with $0 \leq a<b$---is monotone increasing in the range $x \in [0,\infty)$ and $|S| \leq n^{\frac{1}{|p|+1}}$ gives us
\begingroup
\allowdisplaybreaks  
\begin{align}
\left(\frac{\frac{1}{n} \sum_{i=1}^n v_i(\omega_i)^p}{\frac{1}{n}\sum_{i=1}^n v_i(x_i)^p}\right)^{1/p} & \leq 2n \left(\frac{  n^{\frac{1}{|p|+1}} +\left(\frac{\fac}{n}\right)^{|p|} \sum_{i\in S^c}\left(\frac{1}{v_i(\omega_i)}\right)^{|p|}}{ n^{\frac{1}{|p|+1}} +\sum_{i\in S^c}\left(\frac{1}{v_i(\omega_i)}\right)^{|p|}}\right)^{1/|p|} \nonumber \\
& = \left(\frac{(2n)^{|p|}n^{\frac{1}{|p|+1}}+(2\fac)^{|p|}\sum_{i\in S^c}\left(\frac{1}{v_i(\omega_i)}\right)^{|p|}}{n^{\frac{1}{|p|+1}}+\sum_{i\in S^c}\left(\frac{1}{v_i(\omega_i)}\right)^{|p|}}\right)^{1/|p|} \nonumber \\
& \leq 2\fac \left(\frac{n+\sum_{i\in S^c}\left(\frac{1}{v_i(\omega_i)}\right)^{|p|}}{n^{\frac{1}{|p|+1}}+\sum_{i\in S^c}\left(\frac{1}{v_i(\omega_i)}\right)^{|p|}}\right)^{1/|p|} \label{ineq:something}
\end{align}
\endgroup
The last inequality follows from $\frac{(2n)^{|p|}n^{\frac{1}{|p|+1}}}{(2\fac)^{|p|}} = \left(\frac{n}{\fac} \right)^{|p|} n^{\frac{1}{|p|+1}} \leq \left( n^{1 - \frac{|p|}{|p|+1}} \right)^{|p|} \cdot n^{\frac{1}{|p|+1}} = n$; recall that $\fac=8n^{\frac{|p|}{|p|+1}}\log^2(2n)$.

Furthermore, using inequality (\ref{ineq:something}) and the fact that the function $g(y) \coloneqq \frac{a'+y}{b'+y}$---with $a'>b'\geq 0$---is monotone decreasing in the range $y \in [0,\infty)$ and $\sum_{i\in S^c}\left(\frac{1}{v_i(\omega_i)}\right)^{|p|}\geq |S^c| =  n-|S|$, we obtain 
\begin{align*}
\left(\frac{\frac{1}{n} \sum_{i=1}^n v_i(\omega_i)^p}{\frac{1}{n}\sum_{i=1}^n v_i(x_i)^p}\right)^{1/p}  &  \leq 2\fac \left(\frac{n+n-|S|}{n^{\frac{1}{|p|+1}}+n-|S|}\right)^{1/|p|}\\& \leq 2\fac \left(\frac{n+n}{n^{\frac{1}{|p|+1}}+n-|S|}\right)^{1/|p|} \\
& \leq 2\fac  \left(\frac{2n}{n}\right)^{1/|p|}\tag{since $|S|\leq n^{\frac{1}{|p|+1}}$}\\
& =2^{1+1/|p|}\ \fac
\end{align*}  

Hence, $\left( \frac{1}{n} \sum_i v_i(x_i)^p \right)^{1/p} \geq \frac{1}{2^{1+1/|p|}\fac} \left(\frac{1}{n}\sum_i v_i(\omega_i)^p \right)^{1/p}$, i.e., ${\rm M}_p(\x)\geq \frac{1}{2^{1+1/|p|}\fac} {\rm M}_p(\opt)$ with $\fac = 8n^{\frac{|p|}{|p|+1}}\log^2(2n)$. This, overall, establishes the stated competitive ratio for allocation $\x$.
\subsection{Improved Guarantee for ${\pmb{\frac{-1}{4} \leq p \leq \frac{-1}{\log (2n)}}}$}
\label{sec:5.5}
This subsection shows that $\alg(\fac)$, with $\fac = 8(2n)^{2|p|}\log^3(2n)$, achieves a competitive ratio of $O\left(n^{2|p|}\log^3n\right)$ for maximizing $p$-mean welfare with $-1/4 \leq p \leq -\frac{1}{\log (2n)}$.\footnote{This competitive ratio is tight, up to an $n^{|p|}$ factor.} Here, let $\opt=(\omega_i)_i$ denote an allocation with $p$-mean welfare, ${\rm M}_{p}(\opt)$, at least half of the optimal $p$-mean welfare and $v_i(\omega_i) \geq \nicefrac{1}{2n}$, for all $i \in [n]$; such an allocation is guaranteed to exist (Remark \ref{remark:near-opt}). 

We will show that the allocation $\x=(x_i)_i$---computed by $\alg \left( 8(2n)^{2|p|}\log^3(2n) \right)$---satisfies ${\rm M}_{p} (\x) \geq \frac{1}{48(2n)^{2|p|}\log^3(2n)} {\rm M}_{p} (\opt)$ and, hence, obtain the stated competitive ratio for $p$-mean welfare.  

Recall that $S_\fac (\opt) \coloneqq \left\{ i \in [n] : v_i(x_i) < \frac{1}{2\fac} v_i(\omega_i) \right\}$ is the subset of agents that are $(2\fac)$-sub-optimal in $\x$. For notational convenience, throughout this subsection, we will use $S$ for the set $S_\fac(\opt)$ and write $S^c \coloneqq [n]\setminus S$. In addition, considering the set $S$, we define the ratio 
\begin{align*}
    \theta \coloneqq \frac{\sum_{i\in S} v_i(\omega_i)^p}{\sum_{i=1}^n v_i(\omega_i)^p}
\end{align*}

Furthermore, we consider the multiplicative drop, $\gamma \in \mathbb{R}_+$, in the $p$-mean welfare from $M_p(\opt)$ when all the agents in $S$ experience a $\frac{1}{2n}$ factor decrement. Formally, 
\begin{align*}
\gamma \coloneqq \frac{\left(\sum_{i\in S^c} v_i(\omega_i)^p + \frac{1}{(2n)^p} \ \sum_{i\in S} v_i(\omega_i)^p \right)^{1/p}}{\left( \sum_{i=1}^n v_i(\omega_i)^p \right)^{1/p}}
\end{align*}
Note that, $v_i(x_i) \geq \frac{1}{2n} \geq \frac{1}{2n} v_i(\omega_i)$, in particular for all $i \in S$. Hence, lower bounding $\gamma$ will enable us to compare $\sum_{i=1}^n v_i(x_i)^p$ with $\sum_{i=1}^n v_i(\omega_i)^p$ by way of $\sum_{i\in S^c} v_i(\omega_i)^p + \frac{1}{(2n)^p} \ \sum_{i\in S} v_i(\omega_i)^p$. 
\begin{proposition}
\label{proposition:const-neg-p-theta}
For any $p \in \left( \nicefrac{-1}{4}, \frac{-1}{\log (2 n)} \right]$, the ratio $\theta \leq \frac{1}{(2n)^{|p|}\log(2n)}$, i.e., 
\begin{align*}
    \sum_{i\in S} v_i(\omega_i)^p \leq \frac{1}{(2n)^{|p|} \log(2n)} \sum_{i=1}^n v_i(\omega_i)^p.
\end{align*}
\end{proposition}
\begin{proof}
Since $p<0$ and $\frac{1}{2n}\leq v_i(\omega_i)\leq 1$, the inequalities $ 1 \leq v_i(\omega_i)^p \leq (2n)^{|p|}$ hold for all $i \in [n]$.
Summing we obtain  
\begin{equation}
\sum_{i=1}^n v_i(\omega_i)^p \geq n \label{-1/4:st1}
\end{equation}
Since in the context at hand $\fac = 8(2n)^{2|p|}\log^3(2n)$, Lemma \ref{commonlem:2} implies $|S| = |S_\fac (\opt)| \leq \frac{8n \log^2 (2n)}{\fac} = \frac{n}{(2n)^{2|p|}\log(2n)}$. This bound and inequality (\ref{-1/4:st1}) imply 
\begin{align*}
\sum_{i\in S} v_i(\omega_i)^p & \leq \sum_{i\in S}(2n)^{|p|}\\
&= |S|(2n)^{|p|}\\
&\leq \frac{n}{(2n)^{2|p|}\log (2n)}(2n)^{|p|}\\
& = \frac{n}{(2n)^{|p|}\log(2n)}\\
&\leq \frac{1}{(2n)^{|p|}\log(2n) }\sum_{i=1}^n v_i(\omega_i)^p \tag{via inequality (\ref{-1/4:st1})}\\
\end{align*}
\end{proof}
The next proposition builds upon the bound on $\theta$. 
\begin{proposition}\label{proposition:theta-1/4}
For any $p \in \left( \nicefrac{-1}{4}, \frac{-1}{\log (2 n)} \right]$, we have $\left((2n)^{|p|}-1\right)\theta\leq |p|$.
\end{proposition}
\begin{proof}
The bound $\theta\leq \frac{1}{(2n)^{|p|}\log (2n)}$ (Proposition \ref{proposition:const-neg-p-theta}) leads to the stated inequality: 
\begin{align*}
\left((2n)^{|p|}-1\right)\theta&\leq  \frac{\left((2n)^{|p|}-1\right)}{(2n)^{|p|}\log (2n)}\\
&\leq \frac{1}{\log(2n)}\\
&\leq |p| \tag{as $p\leq -\frac{1}{\log (2n)}$}\\
\end{align*}
\end{proof}
Using the last proposition, we now prove that $\gamma \geq \nicefrac{1}{e}$.
\begin{align*}
\gamma &=\left(\frac{\sum_{i\in S^c} v_i(\omega_i)^p + \frac{1}{(2n)^p} \sum_{i\in S} v_i(\omega_i)^p}{\sum_{i=1}^n v_i(\omega_i)^p} \right)^{1/p}\\
 &=\left(\frac{\sum_{i=1}^n v_i(\omega_i)^p -\sum_{i\in S} v_i(\omega_i)^p + \frac{1}{(2n)^p}  \sum_{i\in S} v_i(\omega_i)^p}{\sum_{i=1}^n v_i(\omega_i)^p}\right)^{1/p}\\
&=\left((1-\theta)+\frac{1}{(2n)^p}\theta\right)^{1/p}\\
&=\frac{1}{\left((1-\theta)+(2n)^{|p|}\theta\right)^{1/|p|}}\tag{since $p<0$}\\
&=\frac{1}{\left(1+\left((2n)^{|p|}-1\right)\theta\right)^{1/|p|}}\\
&\geq\frac{1}{\left(1+|p|\right)^{1/|p|}}\tag{since $((2n)^{|p|}-1)\theta\leq |p|$; Proposition \ref{proposition:theta-1/4}}\\
&\geq\frac{1}{e}
\end{align*}
We now proceed to establish the competitive ratio for $p \in \left( \nicefrac{-1}{4}, \frac{-1}{\log (2 n)} \right]$.
\begin{align*}
\sum_{i=1}^n v_i(x_i)^p &= \sum_{i\in S^c} v_i(x_i)^p + \sum_{i\in S} v_i(x_i)^p \\
& \leq \sum_{i\in S^c} \left(\frac{v_i(\omega_i)}{16(2n)^{2|p|}\log^3(2n)}\right)^p+ \sum_{i\in S} v_i(x_i)^p \tag{$v_i(x_i)\geq \frac{v_i(\omega_i)}{16(2n)^{2|p|}\log^3(2n)}$, for all $i\in S^c$, and $p<0$}\\ 
& \leq \sum_{i\in S^c}\left(\frac{v_i(\omega_i)}{16(2n)^{2|p|}\log^3(2n)}\right)^p+ \frac{1}{(2n)^p}\sum_{i\in S} v_i(\omega_i)^p \tag{$v_i(x_i)\geq \frac{1}{2n}>\frac{v_i(\omega_i)}{2n}$ for all $i\in[n]$}\\
& \leq \left(\frac{1}{16(2n)^{2|p|}\log^3(2n)}\right)^{p}\left(\sum_{i\in S^c}(v_i(\omega_i))^p+ \frac{1}{(2n)^p}\sum_{i\in S}(v_i(\omega_i))^p\right)\tag{$p<0$}  \label{-1/4:st2}\\
& = \left(\frac{1}{16(2n)^{2|p|}\log^3(2n)}\right)^{p} \ \gamma^p \sum_{i=1}^n v_i(\omega_i)^p \tag{by the definition of $\gamma$}\\
& \leq \left(\frac{1}{48(2n)^{2|p|}\log^3(2n)}\right)^{p} \sum_{i=1}^n v_i(\omega_i)^p \tag{since $\gamma\geq\frac{1}{{e}}>\frac{1}{3}$} \label{-1/4:st3}
\end{align*}
Here $p<0$ and, hence, exponentiating both sides of the last inequality by $1/p$ (and dividing by $n^{1/p}$), we get $\left( \frac{1}{n} \sum_i v_i(x_i)^p \right)^{1/p} \geq \frac{1}{48(2n)^{2|p|}\log^3(2n)} \left(\frac{1}{n}\sum_i v_i(\omega_i)^p \right)^{1/p}$, i.e., ${\rm M}_p(\x)\geq \frac{1}{48(2n)^{2|p|}\log^3(2n)} {\rm M}_p(\opt)$. This, overall, establishes the stated competitive ratio for allocation $\x$.
\subsection{Improved Guarantee for ${\pmb{\frac{-1}{\log(2n)}\leq p<0}}$}
\label{sec:5.4}
In this subsection we show that $\alg(\fac)$, with $\fac = 32\log^3 (2n)$, achieves a competitive ratio of $O\left(\log^3 n\right)$ for maximizing $p$-mean welfare when $-\frac{1}{\log(2n)}\leq p<0$.\footnote{This competitive ratio is tight, up to a poly-log factor.} Here, let $\opt=(\omega_i)_i$ denote an allocation with $p$-mean welfare, ${\rm M}_{p}(\opt)$, at least half of the optimal $p$-mean welfare and $v_i(\omega_i) \geq \nicefrac{1}{2n}$, for all $i \in [n]$; such an allocation is guaranteed to exist (Remark \ref{remark:near-opt}). 

We will show that the allocation $\x=(x_i)_i$---computed by $\alg \left( 32\log^3 (2n) \right)$---satisfies ${\rm M}_{p} (\x) \geq \frac{1}{192\log^3 (2n)} {\rm M}_{p} (\opt)$ and, hence, obtain the stated competitive ratio for $p$-mean welfare.  

Recall that $S_\fac (\opt) \coloneqq \left\{ i \in [n] : v_i(x_i) < \frac{1}{2\fac} v_i(\omega_i) \right\}$ is the subset of agents that are $(2\fac)$-sub-optimal in $\x$. For notational convenience, throughout this subsection, we will use $S$ for the set $S_\fac(\opt)$ and write $S^c \coloneqq [n]\setminus S$. In addition, considering the set $S$, we define the ratio 
\begin{align*}
    \theta \coloneqq \frac{\sum_{i\in S} v_i(\omega_i)^p}{\sum_{i=1}^n v_i(\omega_i)^p}
\end{align*}

Furthermore, we consider the multiplicative drop, $\gamma \in \mathbb{R}_+$, in the $p$-mean welfare from $M_p(\opt)$ when all the agents in $S$ experience a $\frac{1}{2n}$ factor decrement. Formally, 
\begin{align*}
\gamma \coloneqq \frac{\left(\sum_{i\in S^c} v_i(\omega_i)^p + \frac{1}{(2n)^p} \ \sum_{i\in S} v_i(\omega_i)^p \right)^{1/p}}{\left( \sum_{i=1}^n v_i(\omega_i)^p \right)^{1/p}}
\end{align*}
Note that, $v_i(x_i) \geq \frac{1}{2n} \geq \frac{1}{2n} v_i(\omega_i)$, in particular for all $i \in S$. Hence, lower bounding $\gamma$ will enable us to compare $\sum_{i=1}^n v_i(x_i)^p$ with $\sum_{i=1}^n v_i(\omega_i)^p$ by way of $\sum_{i\in S^c} v_i(\omega_i)^p + \frac{1}{(2n)^p} \ \sum_{i\in S} v_i(\omega_i)^p$. 

\begin{proposition}
\label{proposition:theta-small}
For any $p \in \left( \frac{-1}{\log (2 n)}, 0 \right)$, the ratio $\theta \leq \frac{1}{2\log(2n)}$, i.e., 
\begin{align*}
    \sum_{i\in S} v_i(\omega_i)^p\leq \frac{1}{2\log(2n)}\sum_{i=1}^n  v_i(\omega_i)^p.
\end{align*}
\end{proposition}
\begin{proof}
Since $p<0$ and $\frac{1}{2n}\leq v_i(\omega_i)\leq 1$, the inequalities $ 1 \leq v_i(\omega_i)^p \leq (2n)^{|p|}$ hold for all $i \in [n]$.
Summing we obtain  
\begin{equation}
\sum_{i=1}^n v_i(\omega_i)^p \geq n \label{-1/logn,0:st1}
\end{equation}
Moreover, 
\begin{align}
\sum_{i\in S} v_i(\omega_i)^p \leq \sum_{i\in S}\ (2n)^{|p|} = |S| \ (2n)^{|p|} \leq |S| \ (2n)^{\frac{1}{\log(2n)}}  = |S| \ 2^{\frac{\log (2n)}{\log(2n)}} =2|S| \label{ineq:twoS}
\end{align}
Here, the second inequality following from the fact that $ |p| \leq 1/\log (2n)$. Since in the context at hand $\fac = 32\log^3 (2n)$, Lemma \ref{commonlem:2} implies $|S| = |S_\fac (\opt)| \leq \frac{8n \log^2 (2n)}{\fac} = \frac{n}{4 \log (2n)}$. This bound and inequality (\ref{ineq:twoS}) imply 
\begin{align*}
\sum_{i\in S} v_i(\omega_i)^p \leq \frac{n}{2\log (2n)} \leq \frac{1}{2\log (2n)} \sum_{i=1}^n v_i(\omega_i)^p \tag{via inequality (\ref{-1/logn,0:st1})}
\end{align*}
Therefore, the proposition holds. 
\end{proof}

The next proposition builds upon the bound on $\theta$. 
\begin{proposition}
\label{proposition:thetaP}
For any $p \in \left( \frac{-1}{\log (2 n)}, 0 \right)$, we have $((2n)^{|p|}-1)\theta\leq |p|$.
\end{proposition}
\begin{proof}
The function $f(x) \coloneqq 2^{x}-2x$ is monotone decreasing in the range $x\in [0,1]$. Hence, $f(x) \leq f(0) = 1$ for all $x \in [0,1]$. In particular, setting $x=|p|\log (2n) \leq 1$ we get $2^{|p|\log (2n)}-2|p|\log(2n) \leq 1$. Simplifying this bounds leads to the desired inequality: 
\begin{align*}
|p| & \geq \left(2^{|p|\log (2n)}-1\right)\frac{1}{2\log (2n)} \\
& \geq \left(2^{|p|\log (2n)}-1\right)\theta  \tag{since $\theta \leq \frac{1}{2\log (2n)}$; Proposition \ref{proposition:theta-small}} \\
& = \left((2n)^{|p|}-1\right)\theta.
\end{align*}
\end{proof}

Using the last proposition, we now prove that $\gamma \geq \nicefrac{1}{e}$.
\begin{align*}
\gamma &=\left(\frac{\sum_{i\in S^c} v_i(\omega_i)^p + \frac{1}{(2n)^p} \sum_{i\in S} v_i(\omega_i)^p}{\sum_{i=1}^n v_i(\omega_i)^p} \right)^{1/p}\\
 &=\left(\frac{\sum_{i=1}^n v_i(\omega_i)^p -\sum_{i\in S} v_i(\omega_i)^p + \frac{1}{(2n)^p}  \sum_{i\in S} v_i(\omega_i)^p}{\sum_{i=1}^n v_i(\omega_i)^p}\right)^{1/p}\\
&=\left((1-\theta)+\frac{1}{(2n)^p}\theta\right)^{1/p}\\
&=\frac{1}{\left((1-\theta)+(2n)^{|p|}\theta\right)^{1/|p|}}\tag{since $p<0$}\\
&=\frac{1}{\left(1+((2n)^{|p|}-1)\theta\right)^{1/|p|}}\\
&\geq\frac{1}{\left(1+|p|\right)^{1/|p|}}\tag{since $((2n)^{|p|}-1)\theta\leq |p|$; Proposition \ref{proposition:thetaP}}\\
&\geq\frac{1}{e} 
\end{align*}

This lower bound on $\gamma$ leads us to the desired competitive ratio:
\begin{align*}
\sum_{i=1}^n v_i(x_i)^p &= \sum_{i\in S^c} v_i(x_i)^p + \sum_{i\in S} v_i(x_i)^p \\
& \leq \sum_{i\in S^c} \left(\frac{v_i(\omega_i)}{64\log^3(2n)}\right)^p + \sum_{i\in S} v_i(x_i)^p \tag{$v_i(x_i)\geq \frac{v_i(\omega_i)}{64\log^3(2n)}$, for all $i\in S^c$, and $p<0$}\\
& \leq \sum_{i\in S^c}\left(\frac{v_i(\omega_i)}{64\log^3(2n)}\right)^p+ \frac{1}{(2n)^p}\sum_{i\in S} v_i(\omega_i)^p\tag{$v_i(x_i)\geq \frac{1}{2n}>\frac{v_i(\omega_i)}{2n}$, for all $i$, and $p<0$} \\
& \leq (64\log^3(2n))^{-p}\left(\sum_{i\in S^c} v_i(\omega_i)^p + \frac{1}{(2n)^p}\sum_{i\in S} v_i(\omega_i)^p \right) \tag{$p<0$}\\
& = (64\log^3(2n))^{-p} \ \gamma^p \sum_{i=1}^n v_i(\omega_i)^p\tag{by definition of $\gamma$}\\
& \leq (192\log^3(2n))^{-p} \sum_{i=1}^n v_i(\omega_i)^p \tag{since $\gamma\geq\frac{1}{e}>\frac{1}{3}$}
\end{align*}
Here, $p<0$ and, hence, exponentiating both sides of the last inequality by $1/p$ (and dividing by $n^{1/p}$), we get $\left( \frac{1}{n} \sum_i v_i(x_i)^p \right)^{1/p} \geq \frac{1}{192\log^3(2n)} \left(\frac{1}{n}\sum_i v_i(\omega_i)^p \right)^{1/p}$, i.e., ${\rm M}_p(\x)\geq \frac{1}{192\log^3(2n)} {\rm M}_p(\opt)$. This, overall, establishes the stated competitive ratio for allocation $\x$.

\subsection{Tight Guarantee for ${\pmb{p \in (0, 1]}}$}
\label{sec:5.3}
This subsection shows that $\alg(\fac)$, with $\fac = 16\log^3 (2n)$, achieves a competitive ratio of $O\left(\log^3 n\right)$ for maximizing $p$-mean welfare with $0<p\leq 1$. 

Here, let $\opt=(\omega_i)_i$ denote an allocation with $p$-mean welfare, ${\rm M}_{p}(\opt)$, at least half of the optimal $p$-mean welfare and $v_i(\omega_i) \geq \nicefrac{1}{2n}$, for all $i \in [n]$; such an allocation is guaranteed to exist (Remark \ref{remark:near-opt}). We will show that the allocation $\x=(x_i)_i$---computed by $\alg \left( 16\log^3 (2n) \right)$---satisfies ${\rm M}_{p} (\x) \geq \frac{1}{128\log^3 (2n)} {\rm M}_{p} (\opt)$ and, hence, obtain the stated competitive ratio for $p$-mean welfare.

Recall that $S_\fac (\opt) \coloneqq \left\{ i \in [n] : v_i(x_i) < \frac{1}{2\fac} v_i(\omega_i) \right\}$ is the subset of agents that are $(2\fac)$-sub-optimal in $\x$. For notational convenience, throughout this subsection, we will use $S$ for the set $S_\fac(\opt)$ and write $S^c \coloneqq [n]\setminus S$. We partition the set $S=S_\fac (\opt)$ into $\log (2n)$ subsets, based on the values $v_i(\omega_i)$s; specifically, for each $\alpha \in \left\{ \frac{1}{2^k} : 1 \leq k \leq \log(2n) \right\}$, write subset $S\left( \alpha \right) = \left\{i \in S : \alpha \leq v_i(\omega_i) < 2\alpha \right\}$. Since $\frac{1}{2n} \leq v_j(\omega_j) < 1$, for all agents $j \in [n]$, the subsets $S\left(\alpha\right)$s form a partition of $S$. In addition, considering the set $S$, we define the ratio 
\begin{align*}
    \theta \coloneqq \frac{\sum_{i\in S} v_i(\omega_i)^p}{\sum_{i=1}^n v_i(\omega_i)^p}
\end{align*}
Our derivation of the competitive ratio rests on a case analysis with respect to $\theta$. Specifically, below we consider the following two cases - Case {\rm I}: $\theta \leq \frac{1}{\log(2n)}$ and Case {\rm II}: $\theta > \frac{1}{\log(2n)}$.

Also, for each $\alpha \in \left\{ \frac{1}{2^k} : 1 \leq k \leq \log(2n) \right\}$ and agent $i \in [n]$, write 
\begin{align*}
    v_i(x_i^\alpha) \coloneqq \sum_{t=1}^T x_i^{t,\alpha} v_i^t.
\end{align*}
Here, $x_i^{t,\alpha}$-s are the fractional assignments computed in \alg($\fac$).

\noindent
{\bf Case {\rm I}}:  $\theta \leq \frac{1}{\log(2n)}$ (i.e., $\sum_{i\in S} v_i(\omega_i)^p \leq \frac{1}{\log (2n)} \sum_{i=1}^n v_i(\omega_i)^p$). 

We consider the multiplicative drop, $\gamma \in \mathbb{R}_+$, in the $p$-mean welfare from $M_p(\opt)$ when all the agents in $S$ experience a $\frac{1}{2n}$ factor decrement. Formally, 
\begin{align*}
\gamma \coloneqq \frac{\left(\sum_{i\in S^c} v_i(\omega_i)^p + \frac{1}{(2n)^p} \ \sum_{i\in S} v_i(\omega_i)^p \right)^{1/p}}{\left( \sum_{i=1}^n v_i(\omega_i)^p \right)^{1/p}}
\end{align*}

Note that, $v_i(x_i) \geq \frac{1}{2n} \geq \frac{1}{2n} v_i(\omega_i)$, in particular for all $i \in S$. Hence, lower bounding $\gamma$ will enable us to compare $\sum_{i=1}^n v_i(x_i)^p$ with $\sum_{i=1}^n v_i(\omega_i)^p$ by way of $\sum_{i\in S^c} v_i(\omega_i)^p + \frac{1}{(2n)^p} \ \sum_{i\in S} v_i(\omega_i)^p$.  
Towards this, we prove the following proposition.
\begin{proposition}
In the current case (i.e., $\theta\leq \frac{1}{\log (2n)}$), we have $\gamma \geq \frac{1}{4}$.
\end{proposition}
\begin{proof}
We will first show that $p \geq \left(1-\frac{1}{(2n)^p}\right)\theta$. For this bound, note that the function $f(x) \coloneqq x+2^{-x}$ is monotone increasing for all $x \geq 0$. Hence, for any nonnegative $x$, we have $f(x) \geq f(0) = 1$. In particular, substituting $x=p\log (2n)$ gives us $p\log(2n) + 2^{-p\log (2n)} \geq 1$.\footnote{Recall that in the current context $0 < p \leq 1$.} Simplifying this equation leads to the stated inequality 
\begin{align}
p & \geq \left(1- 2^{-p\log (2n)}\right)\frac{1}{\log (2n)} \nonumber \\
& \geq \left(1- 2^{-p\log (2n)}\right)\theta \tag{since $\theta \leq \frac{1}{\log (2n)}$} \\
& = \left(1- \frac{1}{(2n)^p}\right)\theta \label{ineq:ptheta}
\end{align}
Now, write $\tilde{p} \coloneqq \left(1- \frac{1}{(2n)^p}\right)\theta$ to denote the right-hand-side of inequality (\ref{ineq:ptheta}). Also, note that $\tilde{p} \leq \theta\leq \frac{1}{\log(2n)}\leq \nicefrac{1}{2}$. These observations establish the proposition: 
\begin{align*}
\gamma &=\left(\frac{\sum_{i\in S^c} v_i(\omega_i)^p + \frac{1}{(2n)^p} \sum_{i\in S} v_i(\omega_i)^p}{\sum_{i=1}^n v_i(\omega_i)^p} \right)^{1/p}\\
 &=\left(\frac{\sum_{i=1}^n v_i(\omega_i)^p -\sum_{i\in S} v_i(\omega_i)^p + \frac{1}{(2n)^p}  \sum_{i\in S} v_i(\omega_i)^p}{\sum_{i=1}^n v_i(\omega_i)^p}\right)^{1/p}\\
&=\left((1-\theta)+\frac{1}{(2n)^p}\theta\right)^{1/p}\\
&=\left(1- \left(1- \frac{1}{(2n)^p}\right)\theta\right)^{1/p}\\
&=(1-\tilde{p})^{1/p}\\
&\geq (1-\tilde{p})^{1/\tilde{p}} \tag{since $0< (1-\tilde{p})<1$ and $\tilde{p}\leq p$; inequality (\ref{ineq:ptheta})}\\
& \geq \frac{1}{4} \tag{since $\tilde{p}\leq \frac{1}{2}$}
\end{align*}
\end{proof}
Using this proposition ($\gamma \geq 1/4$), we obtain the competitive ratio for the current case: 
\begin{align}
\sum_{i=1}^n v_i(x_i)^p &= \sum_{i\in S^c} v_i(x_i)^p + \sum_{i\in S} v_i(x_i)^p\nonumber\\
&  \geq \sum_{i\in S^c} \left(\frac{v_i(\omega_i)}{32\log^3(2n)}\right)^p + \sum_{i\in S} v_i(x_i)^p \tag{$v_i(x_i)\geq \frac{v_i(\omega_i)}{32\log^3(2n)}$, for all $i\in S^c$, and $p>0$}\\ 
& \geq \sum_{i\in S^c} \left(\frac{v_i(\omega_i)}{32\log^3(2n)}\right)^p + \frac{1}{(2n)^p}\sum_{i\in S} v_i(\omega_i)^p \tag{$v_i(x_i)\geq \frac{1}{2n}>\frac{v_i(\omega_i)}{2n}$ for all $i\in[n]$}\\
& \geq \frac{1}{\left(32\log^3(2n) \right)^{p}} \left(\sum_{i\in S^c} v_i(\omega_i)^p+ \frac{1}{(2n)^p}\sum_{i\in S} v_i(\omega_i)^p\right) \tag{$p>0$} \\
& = \frac{1}{\left(32\log^3(2n) \right)^{p}} \ \  \gamma^p \sum_{i=1}^n v_i(\omega_i)^p \tag{by the definition of $\gamma$} \\
& \geq \frac{1}{\left(128 \log^3(2n) \right)^{p}} \sum_{i=1}^n v_i(\omega_i)^p \label{c1p>0:st2}\tag{since $\gamma\geq\frac{1}{4}$} 
\end{align}
Here, $0<p \leq 1$ and, hence, exponentiating both sides of the last inequality by $1/p$ (and dividing by $n^{1/p}$), we get $\left( \frac{1}{n} \sum_i v_i(x_i)^p \right)^{1/p} \geq \frac{1}{128\log^3(2n)} \left(\frac{1}{n}\sum_i v_i(\omega_i)^p \right)^{1/p}$, i.e., ${\rm M}_p(\x)\geq \frac{1}{128\log^3(2n)} {\rm M}_p(\opt)$. This, overall, establishes the stated competitive ratio for allocation $\x$ in the case $\theta\leq \frac{1}{\log (2n)}$. \\

\noindent
\textbf{Case {\rm II}:} $\theta > \frac{1}{\log(2n)}$ (i.e., $\sum_{i=1}^n v_i(\omega_i)^p < \log (2n)\sum_{i\in S} v_i(\omega_i)^p$). 

In this case the contribution of the $(2\fac)$-sub-optimal agents (i.e., agents in $S$) in the $p$-mean welfare of $\opt$ is sufficiently high. Here, we bound the size of the sets $S\left( \alpha \right) \coloneqq \left\{i \in S  : \alpha \leq v_i(\omega_i) < 2\alpha \right\}$, which partition $S$ (with $\alpha \coloneqq \left\{ \frac{1}{2^k} : 1 \leq k \leq \log(2n) \right\}$).  Recall that  $v_i(x_i^\alpha) = \sum_{s=1}^T x_i^{s,\alpha} v_i^s$. 

\begin{proposition}\label{plem:3}
For any $\alpha\in \left\{ \frac{1}{2^k} : 1 \leq k \leq \log(2n)\right\}$, the size $|S(\alpha)|\leq \frac{6}{\alpha} \log (2n) \sum_{i=1}^nv_i(x_i^\alpha)$.
\end{proposition}
\begin{proof}
For each agent $i\in S(\alpha) \subseteq S = S_\fac (\opt)$, we have $v_i(x_i) < \frac{v_i(\omega_i)}{2 \fac} < \frac{2\alpha}{2\fac}=\frac{\alpha}{\fac}$. Since $\sum_{s=1}^T x^{s, \alpha}_i v^s_i \leq v_i(x_i)$, we get that agent $i$ continues to be in the active set (for $\alpha$) throughout the execution of the algorithm, i.e., $i \in A^\alpha_T$. In fact, agent $i \in B^\alpha_T$, since the value of the last good $T$ is at most $1/n^2$, which in turn in upper bounded by $\alpha/4$ (see Line \ref{line:btalpha} of \alg($\fac$)).  
Furthermore, for each agent $i \in S(\alpha)$, we have $v_i(\omega_i) \geq \alpha$, i.e., $ i \in H(\alpha ,\opt)$. These observations imply that every agent $ i \in S(\alpha)$ is contained in $B^\alpha_T \cap H(\alpha ,\opt)$ as well; equivalently, $S(\alpha) \subseteq B^\alpha_T \cap H(\alpha ,\opt)$. 

Therefore, inequality (\ref{ineq:social-welf}) (obtained in Lemma \ref{commonlem:1}) gives us 

\begin{align}
    \left|S\left( \alpha \right)\right| \leq |B^\alpha_T \cap H(\alpha ,\opt) | \leq  \frac{6}{\alpha} \log (2n) \sum_{i=1}^n \sum_{s=1}^T x^{s, \alpha}_i v^s_i = \frac{6}{\alpha} \log (2n) \sum_{i=1}^n v_i(x_i^\alpha)    
\end{align} 
\end{proof}

Note that $v_i(x_i) = \sum_{\alpha} v_i(x_i^\alpha)$; throughout, we will sum $\alpha$ across the set $\left\{ \frac{1}{2^k} : 1 \leq k \leq \log(2n)\right\}$. The next two propositions provide supporting technical results. 

\begin{proposition}\label{plem:4}
For $0<p\leq1$ and each agent $i \in [n]$, we have  $\sum_{\alpha} v_i(x_i^\alpha)^p \leq \log^{1-p}(2n) \cdot v_i(x_i)^p$.
\end{proposition}
\begin{proof}
Note that $f(x)=x^p$ is a concave function for $p\in(0,1]$. Hence, 
\begin{align*}
\left(\frac{v_i(x_i)}{\log(2n)}\right)^p &= \left(\frac{\sum_{\alpha} v_i(x_i^\alpha)}{\log (2n)}\right)^p\\
&\geq \frac{\sum_{\alpha} v_i(x_i^\alpha)^p}{\log(2n)} \tag{via Jensen's inequality}\\
\end{align*}
Therefore, the desired inequality holds $\sum_{\alpha} v_i(x_i^\alpha)^p \leq \log^{1-p}(2n) \ v_i(x_i)^p$.
\end{proof}

\begin{proposition}\label{plem:5}
Let $\beta, a_1, a_2, \ldots, a_n \in \mathbb{R}_{\geq 0}$ be non-negative real numbers such that $0\leq a_i\leq\beta<1$, for all $i \in [n]$. Also, let $z = \sum_{i=1}^n a_i$. Then, for any $p \in (0,1]$, we have $\sum_{i=1}^n a_i^p\geq \beta^{p-1} \ z$.
\end{proposition} 
\begin{proof}
We start with the given sequence $a_1, a_2, \ldots, a_n \in \mathbb{R}_{\geq 0}$ and transform it into another sequence $\widehat{a}_1, \widehat{a}_2, \ldots, \widehat{a}_n$ with the same sum ($\sum_{i=1}^n \widehat{a}_i = \sum_{i=1}^n a_i = z$) and the property that all the $\widehat{a}_i$s---besides one---are either zero or $\beta$. Furthermore, we will show that $\sum_{i=1}^n a_i^p \geq \sum_{i=1}^n \widehat{a}_i^p \geq \beta^{p-1} \ z$. This will establish the proposition.   

Consider updating any two numbers $a_i \geq a_j$ in the given sequence $\{a_i\}_i$ by setting $\tilde{a}_i = a_i + \delta$ and $\tilde{a}_j = a_j - \delta$ with $\delta = \min \{ \beta - a_i, a_j \}$. Note that, for $p \in (0,1]$ and any $a\geq b>0$, the function $f(\delta) \coloneqq (a+\delta)^p+(b-\delta)^p$ is non-increasing in the range $\delta \in [0,b]$. Hence, the update gives us $a^p_i + a^p_j \geq \tilde{a}_i^p + \tilde{a}_j^p $ and maintains the sum. Also, by construction,  in the new sequence, $\{\tilde{a}_i, \tilde{a}_j \} \cup \{a_k\}_{k \neq i, j}$, the maximum value is still at most $\beta$. 

A repeated application of such an updated leads to a sequence $\{ \widehat{a}_i \}_i$ such that $\sum_{i=1}^n a_i^p \geq \sum_{i=1}^n \widehat{a}_i^p$ along with $\sum_{i=1}^n \widehat{a}_i = z$ and $0\leq \widehat{a}_i \leq \beta$, for all $i \in [n]$. Furthermore, all the numbers $\widehat{a}_i$s---besides at most one---are either zero or $\beta$. Hence, without loss of generality, assume that $\widehat{a}_i \in \{0, \beta\}$, for all $i < n$, and $\widehat{a}_n =  \widehat{\beta}$, for some $0 \leq \widehat{\beta} \leq \beta$. In this setup, 
\begin{align}
\sum_{i=1}^n \widehat{a}_i^p & =\frac{z-\widehat{\beta}}{\beta}\beta^p+ \widehat{\beta}^p \nonumber \\ 
& = z \ \beta^{p-1}+ \widehat{\beta}^{p}-\widehat{\beta} \ \beta^{p-1} \nonumber \\
&\geq z \ \beta^{p-1} \label{ineq:blah}
\end{align}
The last inequality follows from $\widehat{\beta}^{p}-\widehat{\beta} \ \beta^{p-1} = \widehat{\beta} \left( \widehat{\beta}^{p-1} - \beta^{p-1} \right) \geq 0$; recall that $\widehat{\beta} \leq \beta$ and $p \in (0,1]$. 
Inequality (\ref{ineq:blah}) and the bound $\sum_{i=1}^n a_i^p \geq \sum_{i=1}^n \widehat{a}_i^p$ establish the proposition. 
\end{proof}

Using the bounds on $|S(\alpha)|$s and $v_i(x_i^\alpha)$s, we now establish the competitive ratio. Inequality (\ref{ineq:allagents}) gives us $v_i(x_i^\alpha)<\frac{2\alpha}{\fac}=\frac{\alpha}{8\log^3(2n)}$. Considering $\beta = \frac{\alpha}{8\log^3(2n)}$, $a_i= v_i(x_i^\alpha)$ (for all $i$), $z = \sum_{i=1}^n a_i$, and applying Proposition \ref{plem:5} we get
\begin{equation}
    \sum_{i=1}^n v_i(x_i^\alpha) \leq \frac{\sum_{i=1}^n v_i(x_i^\alpha)^p}{\left(\frac{\alpha}{8\log^3(2n)}\right)^{p-1}} \label{plem:6}
\end{equation}

Finally, we have the bound:
\begin{align*}
\sum_{i=1}^n v_i(w_i)^p &\leq \log (2n)\sum_{i\in S} v_i(w_i)^p \tag{current case}\\
& \leq \log (2n) \sum_{\alpha} \sum_{i\in S(\alpha)} v_i(w_i)^p \tag{$S(\alpha)$s partition $S$}\\
& \leq  \log (2n)\sum_{\alpha}\sum_{i\in S(\alpha)} (2\alpha)^p \tag{since $v_i(w_i)<2\alpha$, for all $i\in S(\alpha)$, and $p>0$}\\
& =  \log (2n)\sum_{\alpha} (2\alpha)^p|S(\alpha)|\\
&  \leq \log (2n)\sum_{\alpha} 2^p\alpha^{p}  \frac{6 \log (2n)}{\alpha} \sum_{i=1}^n v_i(x_i^\alpha) \tag{via Proposition \ref{plem:3}} \\
&   \leq \log (2n)\sum_{\alpha} 2^p\alpha^{p-1}\cdot 6 \log (2n) \frac{\sum_{i=1}^n v_i(x_i^\alpha)^p}{\left(\frac{\alpha}{8\log^3(2n)}\right)^{p-1}} \tag{via inequality (\ref{plem:6})}\\
&= 2^p\cdot 6\log^2 (2n)\cdot (8\log^3(2n))^{p-1}\sum_{i=1}^n\sum_{\alpha} v_i(x_i^\alpha)^p\\
&\leq 2^p\cdot 6\log^{3-p}(2n)\cdot (8\log^3(2n))^{p-1}\sum_{i=1}^n v_i(x_i)^p \tag{via Proposition \ref{plem:4}}\\
&= 2^p\cdot 6\cdot 8^{p-1}\cdot \log^{2p}(2n)\sum_{i=1}^n v_i(x_i)^p\\
& \leq 2^p\cdot 8^{p}\cdot \log^{2p}(2n)\sum_{i=1}^n v_i(x_i)^p
\end{align*}
Here,  $0<p \leq 1$ and, hence, exponentiating both sides of the last inequality by $1/p$ (and dividing by $n^{1/p}$), we get $\left( \frac{1}{n} \sum_i v_i(x_i)^p \right)^{1/p} \geq \frac{1}{16\log^2(2n)} \left(\frac{1}{n}\sum_i v_i(\omega_i)^p \right)^{1/p}$, i.e., ${\rm M}_p(\x)\geq \frac{1}{16\log^2(2n)} {\rm M}_p(\opt)$. This, overall, establishes the stated competitive ratio for allocation $\x$ in the case $\theta> \frac{1}{\log (2n)}$.

Overall, this completes the analysis for $p \in (0,1]$.
\section{Lower Bounds}
\label{section:lower-bounds}
This section complements our positive results by establishing lower bounds for online maximization of $p$-mean welfare.  

\subsection{Sub-Optimality for all ${\pmb{p<1}}$}
\label{sec:6.1}
We provide an adversarial instance to show that, for any $p<1$, there does not exist an online algorithm that computes allocations with optimal  $p$-mean welfare.\footnote{Recall that, for $p=1$, a greedy online algorithm finds an allocation with maximum possible (average) social welfare.}    

\LowerBoundSubOpt*
\begin{proof}
We start with the description of the hard instance for any given online algorithm $\mathcal{A}$. Consider a setting with four agents and five goods. Write $\x = (x_i^t)_{i \in [4], t \in [5]}$ to denote the allocation returned by algorithm $\mathcal{A}$. 

In the first round (i.e., for the first good), agents $1$ and $2$ have value $\nicefrac{1}{2}$, and the remaining agents have value $0$. Write $h_1$ to denote the agent that received the maximum allocation in the first round (i.e., the maximum fraction of the first good), $h_1 \coloneqq \argmax_{i\in\{1,2\}}x_i^1$. Also, let $\ell_1 \coloneqq \{1, 2\} \setminus \{h_1\}$, i.e., $\ell_i$ receives a lower (than $h_1$) fraction of the first good. We will assume, without loss of generality, that each good is assigned only among agents that have nonzero value for the good. 

In the second round (i.e., for the second good), agents $3$ and $4$ have value $\nicefrac{1}{2}$, and the remaining agents have value $0$. For this round, write $h_2 \coloneqq \argmax_{i\in\{3,4\}} x_i^2$ and $\ell_2 \coloneqq \{3, 4 \} \setminus \{h_2\}$. 

In the third round, agents $\ell_1$ and $\ell_2$ have value of $\nicefrac{1}{2}$ and the other agents have value $0$. The fourth good is valued at $\nicefrac{1}{2}$ by agent $h_1$ and the other agents value the good at $0$. Finally, in the fifth round, agent $h_2$ has value $\nicefrac{1}{2}$ and the other agents have a value of $0$.

Note that the values in the third, fourth, and fifth rounds are set (adversarially) by considering the allocations made by the algorithm $\mathcal{A}$ in the first two rounds. Also, by construction, for each agent $i$, we have the scaling $\sum_{t=1}^5 v_i^t=1$.\footnote{We can also ensure that $v_i^t \le 1/n^2$, by just dividing each round into $2/n^2$ identical sub-rounds.}

The matrix $B$ (below) provides an example of such a hard instance. The rows of $B$ correspond to the five rounds (goods) and the columns correspond to the four agents. In this example, $\ell_1=1$, $h_1=2$, $\ell_2=3$, and $h_2=4$. The lower bound example can be extended to general $n$ agents by repeating the five rounds $n/4$ times. In each repetition, we select the agents $4k-3$, $4k-2$, $4k-1$, and $4k$ for the $k$th group of five rounds, and for the remaining agents set the value to be zero. Matrix $A$ highlights such an extension.

\[B:=\begin{bmatrix}
\nicefrac{1}{2} & \nicefrac{1}{2} & 0 &0 \\
0 & 0 & \nicefrac{1}{2} & \nicefrac{1}{2} \\
\nicefrac{1}{2} & 0 & \nicefrac{1}{2} & 0\\
0 & \nicefrac{1}{2} & 0 & 0 \\
0 & 0 & 0 & \nicefrac{1}{2} 
\end{bmatrix},\quad\quad
A := \begin{Large}\begin{bmatrix} 
    B & \dots & 0 & \dots &\text{ }0\\
    \vdots & \ddots &\vdots &  & \vdots\\
    0 &  \dots  & B & \dots &\text{ }0\\
    \vdots& &\vdots & \ddots & \vdots &\\
    0&\dots &0 &\dots & B&\\
    \end{bmatrix}\end{Large}
\]

Hence, for the remainder of the proof, we just focus on such an instance with the four agents. First, we derive a lower bound on the optimal $p$-mean welfare, ${\rm M}_p^*$. In particular, consider the following (offline) allocation: in the first round, the fractional assignment for agent $\ell_1$ is $\nicefrac{3}{4}$ and for agent $h_1$ it is $\nicefrac{1}{4}$. In the second round, the fractional assignments of agents $\ell_2$ and $h_2$ are $\nicefrac{3}{4}$ and $\nicefrac{1}{4}$, respectively. In the third round, the good is divided equally among the agents $\ell_1$ and $\ell_2$. In the fourth and fifth rounds, respectively, agents $h_1$ and $h_2$ receive the entire good. This allocation ensures that every agent achieves a total value of $5/8$. Hence, the optimal $p$-mean welfare ${\rm M}_p^* \geq \nicefrac{5}{8}$.

Now, we upper bound the $p$-mean welfare of the allocation $\x = (x_i)_{i \in [4]}$ returned by the online algorithm $\mathcal{A}$. Note that in an allocation, if a nonzero fraction of a good is allocated to an agent who values the good at $0$, then we can improve the $p$-mean welfare by allocating this fraction to an agent who values the good at $\nicefrac{1}{2}$. Hence, we assume that in allocation $\x$ each good is assigned among the agents who have a value of $\frac{1}{2}$ for the good. Let $a \coloneqq x^1_{\ell_1}$, $b\coloneqq x^2_{\ell_2}$, and $c \coloneqq x^3_{\ell_1}$. Then, $x^1_{h_1} = 1 - a$ and $x^2_{h_2} = 1 - b$. In addition, we have $x^4_{h_1} =1$ and $x^5_{h_2}=1$. The values of the agents can be expressed in terms of these fractional assignments as follows: $v_{\ell_1}(x_{\ell_1}) = \frac{a+c}{2}$, $v_{h_1}(x_{h_1}) = \frac{1}{2} + \frac{1-a}{2}$, $v_{\ell_2} (x_{\ell_2}) = \frac{b}{2}+\frac{1-c}{2}$, and $v_{h_2}(x_{h_2}) = \frac{1}{2}+\frac{1-b}{2}$. By construction, we have $a \leq 1/2$, $b\leq 1/2$, and we assume, without loss of generality, that $c\leq 1/2$. 

The $p$-mean welfare of the computed allocation is ${\rm M}_p \left( \frac{a+c}{2},\frac{1}{2}+\frac{1-a}{2},\frac{b}{2}+\frac{1-c}{2},\frac{1}{2}+\frac{1-b}{2}\right)$. {If $p=-\infty$ (Egalitarian Welfare), we have $M_p(\x)\leq \frac{a+c}{2}\leq \frac{1}{2}<\frac{5}{8}={\rm M}_p^*$. Therefore, for the rest of the proof let us assume that $p\in (-\infty,1)$.}

For any $p\in (-\infty,1)$, $c, c' \in \mathbb{R}_{\geq 0}$, and $0 \leq L < H$, the function $f(\delta) \coloneqq {\rm M}_p(L+ \delta, H-\delta, c, c')$ is strictly increasing in the range $\delta \in \left( 0, \frac{H-L}{2} \right)$ (see Proposition \ref{proposition:pigou-dalton} in Appendix \ref{appendix:lower-bound}). Using this fact repeatedly, we obtain:
\begin{align}
{\rm M}_p\left( \x \right) &= {\rm M}_p\left(\frac{a+c}{2},\frac{1}{2}+\frac{1-a}{2},\frac{b}{2}+\frac{1-c}{2},\frac{1}{2}+\frac{1-b}{2}\right) \nonumber \\ 
& \leq {\rm M}_p\left(\frac{a+c}{2}+\frac{1-2a}{4}, \frac{1}{2}+\frac{1-a}{2}-\frac{1-2a}{4},\frac{b}{2}+\frac{1-c}{2},\frac{1}{2}+\frac{1-b}{2}\right)\label{eq:eq2}\\
& \leq {\rm M}_p\left(\frac{1+2c}{4}, \frac{3}{4},\frac{b}{2}+\frac{1-c}{2}+\left(\frac{1}{4}-\frac{b}{2}\right),\frac{1}{2}+\frac{1-b}{2}-\left(\frac{1}{4}-\frac{b}{2}\right)\right) \label{eq:eq3}\\
& = {\rm M}_p\left(\frac{1+2c}{4},\frac{3}{4},\frac{1}{4}+\frac{1-c}{2},\frac{3}{4}\right) \nonumber \\ 
& \leq {\rm M}_p\left(\frac{1+2c}{4}+\frac{1-2c}{4},\frac{3}{4},\frac{1}{4}+\frac{1-c}{2}-\frac{1-2c}{4},\frac{3}{4}\right)\label{eq:eq5}\\
& = {\rm M}_p\left(\frac{1}{2},\frac{3}{4},\frac{1}{2},\frac{3}{4} \right) \nonumber \\ 
& < {\rm M}_p\left(\frac{1}{2}+\frac{1}{8}, \frac{3}{4}-\frac{1}{8}, \frac{1}{2}+\frac{1}{8}, \frac{3}{4}-\frac{1}{8} \right) = {\rm M}_p\left(\frac{5}{8}, \frac{5}{8}, \frac{5}{8}, \frac{5}{8} \right) =\frac{5}{8} \nonumber \\ 
\end{align}
Here, Inequality \eqref{eq:eq2} follows from $\frac{a+c}{2}\leq \frac{1}{2}+\frac{1-a}{2}$ and $0\le \frac{1-2a}{4} \le \left(\frac{1}{2}+\frac{1-a}{2}- \frac{a+c}{2}\right)/2$.
Inequality \eqref{eq:eq3} utilizes $\frac{b}{2}+\frac{1-c}{2}\leq \frac{1}{2}+\frac{1-b}{2}$ and $0\le \frac{1}{4}-\frac{b}{2} \le \left(\frac{1}{2}+\frac{1-b}{2}-\frac{b}{2}-\frac{1-c}{2}\right)/2$.
Inequality \eqref{eq:eq5} follows from $\frac{1}{4}+\frac{c}{2}\leq \frac{1}{4}+\frac{1-c}{2}$ and $0\le \frac{1-2c}{4} \le \left(\frac{1}{4}+\frac{1-c}{2}-\frac{1+2c}{4}\right)/2$.

Therefore, the $p$-mean welfare of the computed allocation $\x$ is sub-optimal, ${\rm M}_p(\x) < 5/8 = {\rm M}_p^*$. The theorem stands proved. 
\end{proof}

The next section provides a stronger lower bound for $p<0$.

\subsection{Lower Bound for ${\pmb{p<0}}$}
\label{sec:6.2}
This section considers exponent parameters $p<0$ and establishes a lower bound of $\left(2^{-\left( 2 + \nicefrac{2}{|p|} \right)} \cdot n^{\frac{|p|}{2|p|+1}} \right)$ for the competitive ratio of $p$-mean welfare maximization. Note that for any fixed (negative) constant $c <0$ and all $p \leq c$, this lower bound reduces to $O\left(n^{\frac{|p|}{2|p|+1} - \varepsilon} \right)$, for any constant $\varepsilon>0$. Furthermore, the lower bound tends to $n^{1/2}$ as $p \to -\infty$.

\LowerBoundGeneral*
\begin{proof}
We begin by describing the adversarial instance. Let $a \coloneqq \frac{|p|+1}{2|p|+1}$ and choose the number of agents $n$ to be greater than $2^{\frac{2|p|+1}{|p|}}$. This choice ensures that $n^{\frac{|p|}{2|p|+1}}>2$ and, hence, $n^{{a}-\frac{a}{|p|+1}}>2$. This inequality reduces to 

\begin{align}
n^a-n^{\frac{a}{|p|+1}}&> \frac{1}{2} n^a \label{eq:eq10} 
\end{align}

Also, note that 
\begin{align}
1-a+\frac{a}{|p|+1} =1-\frac{|p|}{2|p|+1} =\frac{|p|+1}{2|p|+1} =a \label{eq:eq11} 
\end{align}

Equations \eqref{eq:eq10} and \eqref{eq:eq11} will be useful in the analysis.

We divide the agents into $n^{1-a}$ groups, with each group containing $n^a$ agents. In round $t \in [n^{1-a}]$ (i.e., for the $t$-th good), the agents from $t$-th group have value $\frac{1}{2}$ and all other agents have value $0$.

Now, from each group $t \in [n^{1-a}]$, we select $n^{\frac{a}{|p|+1}}$ agents each of whom has received at most $2n^{-a}$ fraction of the good $t$. We we call such agents as special agents. At least  $n^{\frac{a}{|p|+1}}$ special agents exist within each group $t$. Otherwise, the total fraction of good $t$ allotted to the agents in group $t$ would be at least {$(n^a-n^{\frac{a}{|p|+1}})\frac{2}{n^a}>\frac{n^a}{2}\frac{2}{n^a}=1$} (see equation (\ref{eq:eq10})); this inequality leads to a contradiction, since the total fraction allotted is at most $1$. Write $\mathcal{E}$ to denote the set of all the special agents. The total number of special agents $|\mathcal{E}| = n^{\frac{a}{|p|+1}+1-a}=n^a$ (see equation (\ref{eq:eq11})).

Continuing the instance construction, in the round $(n^{1-a}+1)$, all the special agents have value  $\frac{1}{2}$ and the  remaining agents have value $0$. Let $\mathcal{M} \coloneqq [n] \setminus \mathcal{E}$ denote the set of all the non-special agents. Agents in $\mathcal{M}$ have zero value for the good in round $(n^{1-a}+1)$.
{In the round $n^{1-a}+j+1$, with $1\leq j\leq |\mathcal{M}|$, the $j$-th agent in $\mathcal{M}$ values the good---i.e., the $(n^{1-a}+j+1)$-th good---at $\frac{1}{2}$ and all the other agents value the good at $0$.}
This completes the description of the adversarial instance.

Now, we provide a lower bound on the optimal $p$-mean welfare, ${\rm M}^*_p$, in the instance. Consider an (offline) allocation wherein the good in each round $t \in [n^{1-a}]$ is uniformly distributed among the \emph{special agents} of $t$-th group. Furthermore, for all $1\leq j \leq |\mathcal{M}|$, the good in the round $(n^{1-a}+j+1)$ is allotted entirely to the $j$-th agent in $\mathcal{M}$. In such an allocation every special agent receives a total value of at least $\frac{1}{2} \ \frac{1}{\left(n^{\frac{a}{|p|+1}}\right)}$ and every non-special agent receive a total value of at least $\frac{1}{2}$. Therefore, the following bound holds
\begin{align*}
{\rm M}_p^* & \geq \left(\frac{1}{n}\left( |\mathcal{E}|\left(\frac{1}{2n^{\frac{a}{|p|+1}}}\right)^p+ \left(n-|\mathcal{E}| \right)\left(\frac12\right)^p\right)\right)^{1/p} \\
&= \left(\frac{n}{n^{a}\left(2n^{\frac{a}{|p|+1}}\right)^{|p|}+2^{|p|}\left(n-n^{a}\right)}\right)^{1/|p|} \tag{since $|\mathcal{E}| = n^a$ and $p<0$} \\
&\geq \left(\frac{n}{n^a \ 2^{|p|} \ n^{\frac{a|p|}{|p|+1}} + 2^{|p|} \ n}\right)^{1/|p|}\\
&= \left(\frac{n}{2^{|p|}n+2^{|p|}n}\right)^{1/|p|} \tag{since $a+\frac{a|p|}{|p|+1} =1$}\\
&=\frac{1}{2^{1+\frac{1}{|p|}}}
\end{align*}

Next, we upper bound the $p$-mean welfare of the allocation {$\x = (x_i)_{i\in[n]}$} computed online for the instance. Each special agent at the end of round $n^{1-a}$ received a total value of at most $\frac{1}{2} \ \frac{2}{n^a}=\frac{1}{n^a}$. Furthermore, let $\mathcal{F} \subseteq \mathcal{E}$ denote the subset of special agents each of whom received at most $\frac{2}{|\mathcal{E}|} = \frac{2}{n^a}$ fraction of the good in round $(n^{1-a}+1)$. Note that $|\mathcal{E} \setminus \mathcal{F}| \leq \frac{n^a}{2}$; otherwise, the total fraction of good allotted to the special agents would be greater than $\frac{n^{a}}{2} \ \frac{2}{n^{a}}=1$, which is a contradiction. Since $|\mathcal{E}| =n^a$ and $|\mathcal{E} \setminus \mathcal{F}| \leq \frac{n^a}{2}$, we get $|\mathcal{F}| \geq \frac{n^a}{2}$. Also, the definition of these sets ensure that, for each agent $i \in \mathcal{F}$, the total value $v_i(x_i) \leq \frac{1}{2} \frac{2}{n^a} + \frac{1}{2} \frac{2}{n^a} = \frac{2}{n^a}$. Using these observations we obtain
\begin{align*}
M_p\left( \x \right)&\leq \left(\frac{1}{n}\left(|\mathcal{F}| \cdot \left(\frac{2}{n^a}\right)^p + \left(n-|\mathcal{F}|\right)\cdot 1^p\right)\right)^{1/p} \\
&= \left(\frac{n}{\left(\frac{n^a}{2}\right)^{|p|} \cdot |\mathcal{F}|+n-|\mathcal{F}|}\right)^{1/|p| } \tag{since $p<0$}\\
&\leq \left(\frac{n}{ \left(\frac{n^a}{2}\right)^{|p|} \cdot |\mathcal{F}| } \right)^{1/|p|}\\
&\leq \left(\frac{n}{ \left(\frac{n^a}{2}\right)^{|p|} \  \frac{n^a}{2} } \right)^{1/|p|} \tag{since $|\mathcal{F}| \geq n^a/2$}\\ 
&= 2\left(\frac{2n}{n^{\frac{|p|^2+2|p|+1}{2|p|+1}}}\right)^{1/|p|} \tag{$a=\frac{|p|+1}{2|p|+1}$}\\
&= 2\left(\frac{2}{n^{\frac{|p|^2}{2|p|+1}}}\right)^{1/|p|}\\
& = \frac{2^{1+\frac{1}{|p|}}}{n^{\frac{|p|}{2|p|+1}}}
\end{align*}

These bound establish the stated competitive ratio
\begin{align*}
\frac{{\rm M}^*_p}{{\rm M}_p(\x)} & \geq \frac{\left(\frac{1}{2^{1+\frac{1}{|p|}}}\right)}{\left(\frac{2^{1+\frac{1}{|p|}}}{n^{\frac{|p|}{2|p|+1}}}\right)} = \frac{n^{\frac{|p|}{2|p|+1}}}{2^{2+\frac{2}{|p|}}}.
\end{align*}
\end{proof}

The following corollary directly follows from Theorem \ref{theorem:lower-bound} and provides a lower bound for instances wherein the number of agents, $n$, is sufficiently larger than $1/|p|$. 

\begin{restatable}{corollary}{CorollaryLowerBoundSubLinear}
For any $p<0$, there exists an integer $n_0 \in \mathbb{Z}_+$ such that the $p$-mean welfare maximization problem---with $n\geq n_0$ agents---does not admit an online algorithm with competitive ratio $O \left( n^{\frac{|p|}{2|p|+1} - \varepsilon} \right)$, for any constant $\varepsilon >0$.
\end{restatable}

The next corollary states the lower bound specifically for egalitarian welfare ($p=-\infty$). 
\begin{corollary}
\label{corollary:lower-bound-egalitarian}
For any constant $\varepsilon>0$, there does not exist an online algorithm with competitive ratio $n^{1/2 - \varepsilon}$ for egalitarian welfare maximization. 
\end{corollary}

\section{Conclusion and Future Work}
This work studies online allocation of divisible goods and develops encompassing guarantees for $p$-mean welfare objectives. Our results hold under a standard (in the fair division literature) scaling assumption. Relaxing this assumption by, say, considering the problem in the algorithms-with-prediction framework \cite{mitzenmacher2020algorithms} is an interesting direction for future work. Another relevant direction would be study online $p$-mean welfare maximization with stochastic valuations or in the random-order-arrival model. Connecting approximation guarantees for $p$-mean welfare and other well-studied fairness criteria, such as (bounded) envy, is a meaningful thread as well.  

The current paper focussed on divisible goods. However, some of our results extend to settings wherein the goods cannot be fractionally assigned, i.e., extend to indivisible goods. In particular, under assumption that all the (indivisible) goods have sufficiently small values, one can obtain high-probability bounds for egalitarian welfare. Working with such (beyond worst case) assumptions and studying online $p$-mean welfare maximization for indivisible goods will also be interesting.  

\bibliographystyle{alpha}
\bibliography{references}

\newcommand{\etalchar}[1]{$^{#1}$}
\begin{thebibliography}{CMSM{\etalchar{+}}15}

\bibitem[AAGW15]{aleksandrov2015online}
Martin~Damyanov Aleksandrov, Haris Aziz, Serge Gaspers, and Toby Walsh.
\newblock Online fair division: Analysing a food bank problem.
\newblock In {\em Twenty-Fourth International Joint Conference on Artificial
  Intelligence}, 2015.

\bibitem[ABJ10]{azar2010allocate}
Yossi Azar, Niv Buchbinder, and Kamal Jain.
\newblock How to allocate goods in an online market?
\newblock In {\em European Symposium on Algorithms}, pages 51--62. Springer,
  2010.

\bibitem[AW20]{aleksandrov2020online}
Martin Aleksandrov and Toby Walsh.
\newblock Online fair division: A survey.
\newblock In {\em Proceedings of the AAAI Conference on Artificial
  Intelligence}, volume~34, pages 13557--13562, 2020.

\bibitem[BCE{\etalchar{+}}16]{brandt2016handbook}
Felix Brandt, Vincent Conitzer, Ulle Endriss, J{\'e}r{\^o}me Lang, and Ariel~D
  Procaccia.
\newblock {\em Handbook of computational social choice}.
\newblock Cambridge University Press, 2016.

\bibitem[BEP{\etalchar{+}}19]{blazewicz2019handbook}
Jacek Blazewicz, Klaus Ecker, Erwin Pesch, G{\"u}nter Schmidt, and J~Weglarz.
\newblock {\em Handbook on scheduling}.
\newblock Springer, 2019.

\bibitem[BGGJ21]{banerjee2020online}
Siddhartha Banerjee, Vasilis Gkatzelis, Artur Gorokh, and Billy Jin.
\newblock Online nash social welfare maximization with predictions.
\newblock {\em arXiv e-prints - arXiv:2008.03564}, 2021.
\newblock Accessed on August 25, 2021.

\bibitem[BKPP18]{benade2018make}
Gerdus Benade, Aleksandr~M Kazachkov, Ariel~D Procaccia, and
  Christos-Alexandros Psomas.
\newblock How to make envy vanish over time.
\newblock In {\em Proceedings of the 2018 ACM Conference on Economics and
  Computation}, pages 593--610, 2018.

\bibitem[BMS19]{bogomolnaia2019simple}
Anna Bogomolnaia, Herv{\'e} Moulin, and Fedor Sandomirskiy.
\newblock A simple online fair division problem.
\newblock {\em arXiv preprint arXiv:1903.10361}, 1, 2019.

\bibitem[BT96]{brams1996fair}
Steven~J Brams and Alan~D Taylor.
\newblock {\em Fair Division: From cake-cutting to dispute resolution}.
\newblock Cambridge University Press, 1996.

\bibitem[CMSM{\etalchar{+}}15]{correa2015strong}
Jos{\'e} Correa, Alberto Marchetti-Spaccamela, Jannik Matuschke, Leen Stougie,
  Ola Svensson, V{\'\i}ctor Verdugo, and Jos{\'e} Verschae.
\newblock Strong lp formulations for scheduling splittable jobs on unrelated
  machines.
\newblock {\em Mathematical Programming}, 154(1):305--328, 2015.

\bibitem[DJ12]{devanur2012online}
Nikhil~R Devanur and Kamal Jain.
\newblock Online matching with concave returns.
\newblock In {\em Proceedings of the forty-fourth annual ACM symposium on
  Theory of computing}, pages 137--144, 2012.

\bibitem[ES06]{epstein2006online}
Leah Epstein and Rob~Van Stee.
\newblock Online scheduling of splittable tasks.
\newblock {\em ACM Transactions on Algorithms (TALG)}, 2(1):79--94, 2006.

\bibitem[GP15]{goldman2015spliddit}
Jonathan Goldman and Ariel~D Procaccia.
\newblock Spliddit: Unleashing fair division algorithms.
\newblock {\em ACM SIGecom Exchanges}, 13(2):41--46, 2015.

\bibitem[GPT20]{gkatzelis2020fair}
Vasilis Gkatzelis, Alexandros Psomas, and Xizhi Tan.
\newblock Fair and efficient online allocations with normalized valuations.
\newblock {\em arXiv preprint arXiv:2009.12405}, 2020.

\bibitem[JKMR21]{jansen2021empowering}
Klaus Jansen, Kim-Manuel Klein, Marten Maack, and Malin Rau.
\newblock Empowering the configuration-ip: new ptas results for scheduling with
  setup times.
\newblock {\em Mathematical Programming}, pages 1--35, 2021.

\bibitem[Leu04]{leung2004handbook}
Joseph~YT Leung.
\newblock {\em Handbook of scheduling: algorithms, models, and performance
  analysis}.
\newblock CRC press, 2004.

\bibitem[Mou04]{moulin2004fair}
Herv{\'e} Moulin.
\newblock {\em Fair division and collective welfare}.
\newblock MIT press, 2004.

\bibitem[MV21]{mitzenmacher2020algorithms}
Michael Mitzenmacher and Sergei Vassilvitskii.
\newblock Algorithms with predictions.
\newblock {\em In Beyond the Worst-Case Analysis of Algorithms}, Cambridge,
  2021.

\bibitem[Pin12]{pinedo2012scheduling}
Michael Pinedo.
\newblock {\em Scheduling}, volume~29.
\newblock Springer, 2012.

\bibitem[Pre17]{prendergast2017food}
Canice Prendergast.
\newblock How food banks use markets to feed the poor.
\newblock {\em Journal of Economic Perspectives}, 31(4):145--62, 2017.

\bibitem[ZP20]{zeng2020fairness}
David Zeng and Alexandros Psomas.
\newblock Fairness-efficiency tradeoffs in dynamic fair division.
\newblock In {\em Proceedings of the 21st ACM Conference on Economics and
  Computation}, pages 911--912, 2020.

\end{thebibliography}

\appendix
\section{Missing Proof from Section \ref{section:lower-bounds}}
\label{appendix:lower-bound}

\begin{proposition}
\label{proposition:pigou-dalton}
For any $p\in (-\infty, 1)$, $c, c' \in \mathbb{R}_{\geq 0}$, and $0 \leq L < H$, the function $f(\delta) \coloneqq {\rm M}_p(L+ \delta, H-\delta, c, c')$ is strictly increasing in the range $\delta \in \left( 0, \frac{H-L}{2} \right)$
\end{proposition}
\begin{proof}
We divide the proof into two cases based on the value of $p$

\noindent
\textbf{Case 1:} $p\neq 0$.

In this case we have $f(\delta)=\left(\frac{(L+\delta)^p+(H-\delta)^p+c^p+(c')^p}{4}\right)^{\nicefrac{1}{p}}$. The derivative $f'(\delta)=\frac{(L+\delta)^{p-1}-(H-\delta)^{p-1}}{4\left(\frac{(L+\delta)^p+(H-\delta)^p+c^p+(c')^p}{4}\right)^{1-\frac{1}{p}}}$ is greater than zero in the range $\left( 0, \frac{H-L}{2} \right)$ as $(L+\delta)^p>(H-\delta)^p$. Hence $f(\delta)$ is strictly increasing in the range $\delta \in \left( 0, \frac{H-L}{2} \right)$. \\

\noindent
\textbf{Case 2:} $p=0$ (Nash Social Welfare)

In this case we have $f(\delta)=\left((L+\delta)(H-\delta)cc'\right)^{\nicefrac{1}{4}}$. The derivative $f'(\delta)=\frac{cc'}{4}\frac{H-L-2\delta}{\left((L+\delta)(H-\delta)cc'\right)^{\frac{3}{4}}}$ is greater than zero in the range $\left( 0, \frac{H-L}{2} \right)$ as $H-L>2\delta$. Hence $f(\delta)$ is strictly increasing in the range $\delta \in \left( 0, \frac{H-L}{2} \right)$.
\end{proof}

\end{document}